\documentclass[review,hidelinks,onefignum,onetabnum]{siamart220329}
\nolinenumbers
\usepackage{lipsum}
\usepackage{amsfonts}
\usepackage{graphicx}
\usepackage{epstopdf}
\usepackage{algorithmic}
\usepackage{subcaption}
\usepackage{wrapfig}
\usepackage{bm, mathtools, nicematrix, blkarray, tikz}
\usepackage{xcolor, cancel, booktabs}    
 \usepackage[most]{tcolorbox}
\usepackage[capitalize,nameinlink]{cleveref}
\usepackage[acronyms, shortcuts]{glossaries}
\usepackage{siunitx}
\usepackage{amssymb, multicol, empheq, bm, bbm}
\usepackage{enumitem}
\usetikzlibrary{matrix,positioning,arrows.meta,calc,fit,quotes, shapes, patterns, decorations.pathreplacing, shadows}

\usepackage{multirow}
 \usepackage{siunitx} 
\ifpdf
  \DeclareGraphicsExtensions{.eps,.pdf,.png,.jpg}
\else
  \DeclareGraphicsExtensions{.eps}
\fi

\newsiamremark{remark}{Remark}
\newsiamremark{assumption}{Assumption}
\newsiamremark{hypothesis}{Hypothesis}
\crefname{hypothesis}{Hypothesis}{Hypotheses}
\newsiamthm{claim}{Claim}

\headers{Breaking Exponential Complexity in GOOP}{
D. Lee, J. Li, L. Peters, G. Bakirtzis, and D. Fridovich-Keil
}

\title{Breaking Exponential Complexity in Games of Ordered Preference: A Tractable Reformulation\thanks{
This is a preprint manuscript. 
\funding{D.H. Lee, J. Li, and D. Fridovich-Keil are partially supported by the Army Research Laboratory under Cooperative Agreements W911NF-23-2-0011 and W911NF-25-2-0021, and by the National Science Foundation under Grants 2211548 and 2336840. G. Bakirtzis is  partially supported by the academic and research chair
\emph{Architecture des Systèmes Complexes} through the following partners: Dassault Aviation, Naval
Group, Dassault Systèmes, KNDS France, Agence de l'Innovation de Défense, and
Institut Polytechnique de Paris.}}}

\author{
Dong Ho Lee\footnotemark[2]\thanks{Aerospace Engineering and Engineering Mechanics, University of Texas at Austin, Austin, TX
78712-1221 USA 
(\email{leedh0124@utexas.edu},
\email{dfk@utexas.edu}.
}
\and
Jingqi Li\footnotemark[3]\thanks{Oden Institute for Computational Engineering and Sciences, University of Texas at Austin, Austin, TX 78712-1221 USA 
(\email{jingqi.li@austin.utexas.edu}.}
\and
Lasse Peters\footnotemark[4]\thanks{Department of Mechanical Engineering, UC Berkeley, Berkeley, CA 94720-1740 USA 
(\email{lasse.peters@mailbox.org)}.}
\and 
Georgios Bakirtzis\footnotemark[5]\thanks{LTCI, Télécom Paris, Institut Polytechnique de Paris, France 
(\email{bakirtzis@telecom-paris.fr)}.}
\and 
David~Fridovich-Keil\footnotemark[2] \footnotemark[3]
}

\usepackage{amsopn}

\newcommand{\argmin}{\mathop{\rm arg\,min}}

\newcommand{\st}{\mathrm{s.t.}~}

\newcommand{\var}{z}

\newcommand{\numplayers}{N}
\newcommand{\numlevels}{K}

\newcommand{\cost}[1]{J^{#1}}

\newcommand{\equality}{h}
\newcommand{\equaldim}{m_{\mathcal{E}}}
\newcommand{\inequality}{g}
\newcommand{\inequaldim}{m_{\mathcal{I}}}

\newcommand{\Lagrangian}{\mathcal{L}}
\newcommand{\R}{\mathbb{R}}
\newcommand{\X}{\mathcal{Z}}

\newcommand{\klevel}[1]{#1^{\mathrm{th}}\text{-level}}

\newacronym[longplural=open-loop Nash equilibria,plural=OLNE]{olne}{OLNE}{open-loop Nash equilibrium}
\newacronym[longplural=games of ordered preference,plural=GOOPs]{goop}{GOOP}{game of ordered preference}
\newacronym{micp}{MiCP}{mixed complementarity problem}
\newacronym{kkt}{KKT}{Karush-Kuhn-Tucker}
\newacronym{mpec}{MPEC}{mathematical program with equilibrium constraints}
\newacronym{mpcc}{MPCC}{mathematical program with complementarity constraints}
\newacronym{gnep}{GNEP}{generalized Nash equilibrium problem}
\newacronym[longplural=generalized Nash equilibria,plural=GNEs]{gne}{GNE}{generalized Nash equilibrium}
\newacronym{nlp}{NLP}{nonlinear programming}
\newacronym{mfcq}{MFCQ}{Mangasarian-Fromowitz constraint qualification}
\newacronym{licq}{LICQ}{linear independence constraint qualification}
\newacronym{cq}{CQ}{constraint qualification}

\crefformat{equation}{(#2#1#3)}
\Crefname{algocfline}{Algorithm}{Algorithms}
\Crefname{algocf}{line}{lines}
\Crefname{assumption}{Assumption}{Assumptions}
\crefrangeformat{assumption}{Assumptions~#3#1#4-#5#2#6}

\newif\ifdissertation
\dissertationfalse % default: journal
\ifpdf
\hypersetup{
  pdftitle={Breaking Exponential Complexity in Games of Ordered Preference: A Tractable Reformulation},
  pdfauthor={Dong Ho Lee, Jingqi Li, Lasse Peters, Georgios Bakirtzis, David Fridovich-Keil}
}
\fi

\begin{document}

\maketitle
% REQUIRED
\begin{abstract}
Games of ordered preference (GOOPs) model multi-player equilibrium problems in which each player maintains a distinct hierarchy of strictly prioritized objectives. % over a shared \david{is it shared, really? they may each have their own strategy space, no?} strategy space. 
Existing approaches solve GOOPs %\david{no need to redefine so soon! also plz rm the link} 
by deriving and enforcing the necessary optimality conditions that characterize lexicographically constrained Nash equilibria through a single-level reformulation.
However, the number of primal and dual variables in the resulting  \ac{kkt} system grows exponentially with the number of preference levels, leading to severe scalability challenges.
%and conditioning \david{why mention conditioning?} challenges. 
We derive a compact reformulation of these necessary conditions that preserves the essential primal stationarity structure across hierarchy levels, yielding a ``reduced'' KKT system whose size grows polynomially with both the number of players and the number of preference levels. 
The reduced system constitutes a relaxation of the complete KKT system, yet it remains a valid necessary condition for local GOOP equilibria. 
For GOOPs with quadratic objectives and linear constraints, we prove that the primal solution sets of the reduced and complete KKT systems coincide.
% KKT system is solution-set equivalence with the full KKT system. 
More generally, for GOOPs with arbitrary (but smooth) nonlinear objectives and constraints, the reduced KKT conditions recover all local GOOP equilibria but may admit spurious non-equilibrium solutions. 
% To eliminate such solutions, we introduce a second-order sufficient condition that restores primal solution exactness.
We introduce a second-order sufficient condition to certify when a candidate point corresponds to a local GOOP equilibrium.
We also develop a primal-dual interior-point method for computing a local GOOP equilibrium with local quadratic convergence. The resulting framework enables scalable and efficient computation of GOOP equilibria beyond the tractable range of existing exponentially complex formulations.
\end{abstract}

% REQUIRED

% REQUIRED

% \input{siam_tex/introduction.tex}
% \input{siam_tex/related_works.tex}
% \input{siam_tex/feedback_stackelberg_equilibrium}
% \input{siam_tex/kkt_conditions.tex}
% \input{siam_tex/LQ_games.tex}
% \input{siam_tex/nonlinear_games.tex}
% \input{siam_tex/experiments.tex}
% \input{siam_tex/conclusions.tex}

% \appendix
% \input{siam_tex/appendix}

\section{Introduction}
% \begin{itemize}
%     \item Motivation: hierarchical multi-objective decision making in optimization and games.
%     \item Limitations of classical/naive/original reformulations: exponential growth and computational intractability (due to high condition number)
%     \item Contribution 1: propose a reduced (compact) formulation of GOOP that is provably equivalent (with respect to primal variables) to the original but exhibits polynomial growth in system size. 
%     \item Contribution 2: equivalence proof for (i) quadratic games and (ii) their nonlinear generalizations.
%     \item Contribution 3: develop primal-dual interior-point based solver and performed convergence analysis 
% \end{itemize}
% stewart2021survey
Making decisions under strictly ordered, potentially conflicting, objectives is a central problem in optimization.
In many applications---ranging from autonomous driving \cite{zanardi2021urban} and power systems \cite{camargo2019lexicographic} to supply chains \cite{YUE201781}---higher-priority objectives must be satisfied before lower-priority ones are considered.
Lexicographic optimization \cite{isermann1982linear,camargo2019lexicographic,cococcioni2018lexicographic} provides a natural framework for such problems by imposing a strict priority ordering among multiple objectives for a single decision-maker.
In contrast, classical multi-objective optimization \cite{hwang2012multiple,gunantara2018review,ehrgott2005multicriteria,miettinen1999nonlinear} studies competing objectives for a single decision-maker through Pareto-optimality \cite{ngatchou2005pareto,jones2022multi} and methods such as weighted sums \cite{fishburn1967additive,gunantara2018review}, $\epsilon$-constraint formulations \cite{haimes1971bicriterion}, and goal programming \cite{charnes1955optimal,jones2022multi}. 
In these approaches, improving one objective may degrade another, and no intrinsic priority structure is imposed.

We consider the multi-agent extension of the lexicographic setting, namely \emph{games of ordered preference} (GOOPs) \cite{RAL-GOOP,11423775}, in which each player solves a lexicographically constrained, coupled optimization problem.
% Many practical systems, however, involve multiple strategically coupled decision-makers whose objectives are interdependent.
GOOPs inherit aspects of noncooperative game theory \cite{bacsar1998dynamic,facchinei2003finite,debreu1952social}, but differ fundamentally from standard generalized Nash games~\cite{di2022newton,laine2023computation,li2024computation}: each player faces a hierarchy of ordered objectives rather than a single objective. 
This hierarchical structure makes the computation of equilibrium solutions substantially more challenging, since a GOOP equilibrium must encode both coupling across players and lexicographic optimality within each player’s problem.

A natural route for modeling and solving GOOPs is through multilevel optimization \cite{colson2007overview,lu2016multilevel,sato2021gradient,shafiei2024trilevel}, where lower-level (high-priority) optimality conditions define upper-level (low-priority) feasibility. 
The existing solution strategy \cite{RAL-GOOP} adopts this perspective and recursively replaces lower-level problems with their KKT conditions, thereby obtaining a single-level \ac{mpcc} \cite{scheel2000mathematical,luo1996mathematical,outrata2013nonsmooth}. 
However, the \ac{mpcc} reformulation introduces lower-level dual variables that become primal variables at the upper levels. 
The expanding set of primal variables causes the number of variables and conditions to grow exponentially with the number of preference levels, and severely limits scalability.
By contrast, lexicographic optimization enforces higher-priority optimality by treating objectives sequentially and constraining each subproblem so as not to degrade higher-priority, objective values. 
However, the principle in lexicographic optimization does not extend readily to noncooperative equilibrium settings, where each player’s feasible set and objectives depend on the decisions of the others.

In this paper, we show that the exponential growth of the existing GOOP reformulation is not intrinsic to the equilibrium conditions, but rather a consequence of recursively flattening the problem hierarchy. We derive a \emph{reduced} KKT system that preserves the nested primal stationarity structure while avoiding redundant dual variable propagation. The proposed system grows polynomially in the number of players and preference levels.
% In this paper, we show that this exponential growth is not inherent to GOOP optimality, but rather an artifact of recursively flattening the hierarchy.
% We proopose a compact \emph{reduced} KKT system that preserves the essential nested primal stationarity structure while avoiding redundant dual propagation. 
% The resulting formulation grows polynomially with the number of players and preference levels.
Our contributions are threefold. 
We first show that the reduced system is a relaxation of the complete KKT system.
Second, for quadratic GOOPs with linear constraints, we prove primal solution equivalence between the reduced and complete KKT systems. For GOOPs with arbitrary (but smooth) nonlinear objectives and constraints, we show that the reduced KKT system may potentially admit additional solutions. 
Accordingly, we establish a second-order sufficient condition that certifies local GOOP equilibria.
Finally, we develop a primal-dual interior-point method for the reduced KKT system and prove its local quadratic convergence.

\section{Games of Ordered Preference}
% \begin{itemize}
%     \item Notation and Preliminaries
%     \item Definition of the Game of Ordered Preference: nested optimization structure, discuss the hierarchical decision interpretation and comparison to lexicographic optimization
%     \item GOOP Equilibrium Concept
%     \item First-Order Conditions for Local GOOP Equilibria
%     \begin{itemize}
%         \item Assumptions: differentiability, LICQ and strict complementarity
%         \item Derive KKT conditions for each level $k$ and highlight how dual variables propagate upward. 
%         \item Present the Original GOOP KKT System explicitly. 
%     \end{itemize}
% \end{itemize}

In this section, we formalize \acp{goop}, define the local generalized Nash solution concept, and derive corresponding first-order necessary conditions.
% introduce the mathematical formulation of \acp{goop}. We first introduce the concept of local generalized Nash equilibrium of \ac{goop} and then present the necessary conditions for such equilibria. 
% \yorgos{you don't need to define goop every time, just in the abstract and introduction, then refer to the abbreviation}.

% \paragraph{Notation.}
% Let $[\numplayers]$ denote the set $\{1, ..., \numplayers\}$ and $\horizon$ be a positive integer that represents the time horizon over which the game is played.
% We use boldface to denote a time-indexed vector of length $\horizon$ and superscript $i$ to denote the player $i \in [\numplayers]$.
% To this end, we let $\var^i \coloneqq [z^i_1, z^i_2, ..., z^i_\horizon]$ to be a vector of length $n_i \coloneqq n^i_z \times T$ that represents player $i$'s variables. 
% With $\numplayers$ players, the dimension of the game is $n \coloneqq \sum_{i=1}^\numplayers n_i$.
% We denote $\var \coloneqq [\var^1, ..., \var^2]$ (of length $n$) as the concatenation of all players' variables.  
% We use negation to denote the inclusion of all agents except itself, \ie, $\var^{\neg i} \coloneqq \var \setminus \var^i$, which is of length $n - n_i$.

\subsection{Local Generalized Nash Equilibrium for GOOP}
% \subsection{Mathematical formulation of \acp{goop}}
We consider a non-cooperative \ac{goop} involving $\numplayers$ players. 
Let $[\numplayers]:=\{1, 2, ..., \numplayers\}$.
For a player $i \in [\numplayers]$, we denote their decision variables as a real-valued vector $\var^i \in \mathbb{R}^{n^i}$. 
We define the collection of all players' variables as $\var \coloneqq \{\var^1,\dots,\var^\numplayers\}$, and the collection of all players' variables except those of player $i$ as $\var^{\neg i} \coloneqq \{\var^j\}_{j\neq i}$.
When convenient, we identify $\var$ and $\var^{\neg i}$ as vectors, i.e., \(\var = [(\var^1)^\top, \dots, (\var^\numplayers)^\top]^\top \in \mathbb{R}^{n}\), where \(n \coloneqq \sum_{i=1}^\numplayers n^i\), and $\var^{\neg i} = [(\var^j)^\top]^\top_{j\neq i}\in \mathbb{R}^{n-n^i}$.
% For notational convenience, we use the following shorthand for a collection of players' variables: 
% \begin{align}
%     \var &\coloneqq \{\var^1,\dots,\var^\numplayers\} = \{\var^i, \var^{\neg i}\}\\
%     \var^{\neg i} &\coloneqq \{\var^j\}_{j\in[\numplayers]\setminus\{i\}} = \{\var^1,\dots,\var^{i-1},\var^{i+1},\dots,\var^\numplayers\}
% \end{align}
% We interpret these either as collections of player variables or, when convenient, as vectors: e.g., \(\var \coloneqq [(\var^1)^\top, \dots, (\var^\numplayers)^\top]^\top \in \mathbb{R}^{n}\), where \(n \coloneqq \sum_{i=1}^\numplayers n^i\).
Each player has a hierarchy of strictly prioritized objectives, which we refer to as \emph{preferences}.
% \footnote{A prioritized constraint in the form of $f^i_k(\var^i_{k}, \var^{\neg i}_1) \geq 0$ can be encoded as an objective of the form $\cost{i}_{k}(\var^i_{k}, \var^{\neg i}_1) \coloneqq \max\big(0, -f^i_k(\var^i_{k}, \var^{\neg i}_1)\big)$, which can be made smooth via a standard slack variable transformation; cf. \cite{RAL-GOOP,nocedal2006numerical}.}  
We denote the number of preference levels for player $i$ as $\numlevels^i \in \mathbb{N}$.
Unless stated otherwise, players may have different numbers of preference levels, i.e., $\numlevels^i \neq \numlevels^j$ for $i\neq j$.
% We index levels by $k \in [\numlevels^i]$ and adopt the convention that larger $k$ corresponds to higher priority (i.e., lower in the hierarchy). 
% For each player $i$, we index the hierarchy by $k \in \{1, \dots, K_i\}$, with increasing $k$ corresponding to inner levels of the hierarchy and hence higher priority.
% For a given player $i$, we write $\var^i_k\in\mathbb{R}^{n^i}$ for player $i$'s variable at any preference level $k$. 
% Likewise, we denote the other player's variables at level $k$ by \(\var_k^{\neg i}\coloneqq\{\var_k^j\}_{j\in[\numplayers]\setminus\{i\}}\). 
% Each player $i$’s variables at level $k$ are common among all levels.
% Thus, we often drop the level index at the top level and write $\var^i \coloneqq \var^i_1$. 
% Furthermore, each player's preference at level $k$ is characterized by a real-valued function $\cost{i}_k(\var_k^i,\var^{\neg i}),$ where \(\var^{\neg i}\) collects the other players’ top-level variables.
% Throughout the remainder of this paper, when focusing on player $i$’s problem, we implicitly treat the other players’ variables \(\var^{\neg i}\) as fixed and regard them as parameters.
We denote the primal variable of player $i$ at preference level $k$ as $\var_k^i \in \mathbb{R}^{n^i}$.
Player $i$'s decision variable is common to all preference levels; throughout, we may suppress the level index for the top level ($k=1$) and define $\var^i \coloneqq \var_1^i$ to denote player $i$'s top-level decision variable. Each preference level $k$ is associated with a real-valued objective function
$
\cost{i}_k(\var_k^i, \var^{\neg i}),
$
where $\var^{\neg i}$ denotes the collection of the other players’ top-level variables. 
% \david{can axe this sentence} When analyzing player $i$’s problem, we hold $\var^{\neg i}$ fixed.
% % and treat it as a parameter.
% \jingqi{Define $\arg \min$ as local minimum }
Using $\arg \min$ to denote the argument of \emph{local} minimizers, we state player $i$’s problem in the \ac{goop} framework:
\begin{subequations}
\label{eqn:goop-K-level}
\begin{align}
    \label{eqn:goop-K-level-1}
    \quad \min_{\var^i_1} \quad & \cost{i}_1(\var^i_1, \var^{\neg i}) \\
    \label{eqn:goop-K-level-2}
    \st \quad & \var^i_1 \in \argmin_{\var^i_{2}} \: \cost{i}_{2}(\var^i_{2}, \var^{\neg i}) \\
    % \label{eqn:goop-K-level-3}
    & \phantom{\st\quad\qquad} \ddots \nonumber \\
    \label{eqn:goop-K-level-3}
    & \phantom{\st\qquad\,\,} \st \quad \var^i_{\numlevels^i - 1} \in \argmin_{\var^i_{\numlevels^i}} \: \cost{i}_{\numlevels^i}(\var^i_{\numlevels^i}, \var^{\neg i}) \\
    \label{eqn:goop-K-level-4}
    & \qquad\phantom{\st\qquad\st\qquad\qquad\quad} \st \quad \var^i_{\numlevels^i} \in \X_{\numlevels^i}^i(\var^{\neg i}).
\end{align}
\end{subequations}
The innermost feasible set (at level $\numlevels^i$) is given by 
\begin{align}
\label{eqn:Z^i-feasible-set}
\X_{\numlevels^i}^i(\var^{\neg i}) \coloneqq \big\{\,\var^i \;\big|\; \equality^i(\var^i,\var^{\neg i})=\bm{0},\ \inequality^i(\var^i,\var^{\neg i})\ge \bm{0} \big\},
\end{align}
where $\bm{0}$ denotes a vector (or matrix) of all zeros of appropriate dimension, \(\equality^i(\var^i,\var^{\neg i})\in \R^{\equaldim^i}\) and \(\inequality^i(\var^i,\var^{\neg i})\in \R^{\inequaldim^i}\) are equality and inequality constraints, respectively.
% The notation $\argmin$ in \cref{eqn:goop-K-level} is interpreted as the set of \emph{local} minimizers \david{axe. we totally consider global minimizers in the quadratic case. the place to talk about locality is in the (local) equilibrium definition}.

In this paper, we consider the following assumption: 
% david{unclear why you are citing here} \cite{nocedal2006numerical,chinchilla2024newton}:
\begin{assumption}
\label{assum:cost-feasible-set}
% For each player $i \in [\numplayers]$ and each level $k \in [\numlevels^i]$, 
% the objective function $\cost{i}_k(\var^i, \var^{\neg i})$ is continuously differentiable \david{we need $C^{K^i}$ (right?) which is stronger than continuous differentiability or $C^1$. i guess technically we only need high-order differentiability wrt $z^i$ not $z^{\neg i}$}, 
For each player $i \in [\numplayers]$ and each level $k \in [\numlevels^i]$, we assume the objective function $\cost{i}_k(\var^i,\var^{\neg i})$ is sufficiently differentiable with respect to $\var^i$, but not necessarily convex in $\var^i$.
The innermost feasible set $\X_{\numlevels^i}^i(\var^{\neg i})$
% , defined by the nonlinear constraints $\equality^i(\var^i, \var^{\neg i}) = \bm{0}$ and $\inequality^i(\var^i, \var^{\neg i}) \ge \bm{0}$, 
is compact, and an \ac{mpcc}-tailored constraint qualification, such as \ac{mpcc}-\ac{licq} \cite{scheel2000mathematical}, holds.
% Slater's condition \cite{nocedal2006numerical}~holds \david{why do we need Slater specifically? Also cite a specific definition in the book}.
\end{assumption}

Under the above assumption, we formally define the \emph{local generalized Nash equilibrium} for \acp{goop}. 
This definition extends the local generalized Nash equilibrium concept in \cite{facchinei2003finite,palafox2024smooth} to the lexicographic setting.

\begin{definition}[Local Generalized Nash Equilibrium for \acp{goop}]
\label{def:goop-gnep}
% A \ac{goop} is defined by \(\numplayers\) coupled optimization problems
% \jingqi{replace $J^i$ with $J^i_1$ for GOOP nash equilbrium} 
% \david{I feel like you don't need (3.5) and you can write (3.6) directly in reference to (3.3). It could help to define $\X_k^i(z^{\neg i})$ inline in (3.3b,c)}
% \begin{equation}
% \label{eqn:goop-gnep}
% \forall i \in [\numplayers]\quad
% \left\{
% \begin{alignedat}{2}
% &\min_{\var^i}   && \cost{i}_1(\var^i,\var^{\neg i}) \\
% &\text{s.t.}\quad&& \var^i \in \X_1^i(\var^{\neg i}),
% \end{alignedat}
% \right.
% \end{equation}
% where $\X_1^i(\var^{\neg i})$ represents the feasible set of player $i$'s top-level problem and is characterized by the lower-level subproblems \cref{eqn:goop-K-level-2,eqn:goop-K-level-3,eqn:goop-K-level-4}. 
A vector $\var^*$ is a local generalized Nash equilibrium for \ac{goop} \cref{eqn:goop-K-level}, if for each player $i$, the strategy $\var^{i*}$ is feasible for player $i$'s nested problem
% where $\X_k^i(\var^{\neg i})$ represents the feasible set of player $i$'s $\klevel{k}$ problem \cref{eqn:goop-K-level-2,eqn:goop-K-level-3,eqn:goop-K-level-4} 
and there exists an $\epsilon>0$ such that 
\begin{equation}
\cost{i}_1(\var^{i*},\var^{\neg i*})
\;\le\;
\cost{i}_1(\var^i,\var^{\neg i*}),
\quad \forall \var^i \in\X_1^i(\var^{\neg i*})\ \text{s.t.} \  \| \var^i - \var^{i*} \|_2\le \epsilon,\ \forall i\in[\numplayers], 
\end{equation}
where $\X_1^i(\var^{\neg i*})$ denotes the top-level feasible set that encodes the lexicographic preferences of player $i$. Equivalently, no player can unilaterally locally deviate from $\var^{i*}$ and achieve a better top-level objective given $\var^{\neg i*}$ while remaining consistent with their hierarchy of preferences.
% \dongho{Do we need per level version of (3.5)?}\jingqi{Check Proposition A.1. We don't need and it complicates the story}
\end{definition}

In the rest of this work, we refer to a local generalized Nash equilibrium of a GOOP simply as a local GOOP equilibrium.
In the next subsection, building on \cite{RAL-GOOP}, we review the complete KKT conditions for a local GOOP equilibrium and show that its size grows exponentially with hierarchy depth.
% their exponential growth in the number of variables and constraints. 
% Although this notion is well-defined, the formulation in \cref{eqn:goop-K-level} is not directly amenable to computation.
% \david{I feel like it would be good to swap the previous two sentences, logically speaking.} 
This exponential complexity motivates the construction of the reduced KKT system, presented in \Cref{sec:reduced-kkt-formulation}, whose size grows polynomially in the number of players and preference levels.

\subsection{Complete Necessary Conditions for GOOP Equilibria}
\label{sec:complete-kkt-system}
% \todo{Mention the complexity/size growth, in the experiment section, quote the actual number of variables, etc. in the reduced and original KKT system}
The reformulation in \cite{RAL-GOOP} replaces each lower-level problem in \cref{eqn:goop-K-level} with its \emph{complete} \ac{kkt} optimality conditions; we refer to this reformulation as the complete KKT system.
% The resulting reformulation is referred to as the complete KKT system---in contrast to the ``reduced'' KKT system
% whose introduction in \Cref{sec:reduced-kkt-formulation} is a key contribution of this work.
Construction proceeds by backward induction for each player, starting from the innermost level (highest priority) and recursively deriving the KKT systems for outer levels.
For clarity, we denote quantities of the complete KKT system with an overbar to distinguish them from terms in the reduced KKT system introduced in \Cref{sec:reduced-kkt-formulation}.
% , and we identify its associated variables with an overbar, $\bar{\square}$.

% \paragraph{Notation reference for dual variables}
% At each level $k$, the complete KKT reformulation introduces dual variables
% \[
% \bar\lambda_k^i,\ \bar\gamma_k^i,\ \bar\psi_k^i,\ \bar\phi_k^i,
% \]
% with the following roles. 
% The multiplier $\bar\lambda_k^i$ collects \emph{equality-type} multipliers, including those for the player constraints $\equality^i(\bar\var^i,\bar\var^{\neg i})=0$ and the stationarity conditions with respect to lower-level multipliers (i.e., the constraints $\nabla_{\bar\eta_{k+1}^i}\bar{\Lagrangian}_k^i=0$ that arise when $\bar\eta_{k+1}^i$ is treated as an induced primal). 
% The multiplier $\bar\gamma_k^i$ is associated with the (shared/private) inequality constraints $\inequality^i(\bar\var^i,\bar\var^{\neg i})\ge 0$ and appears in the complementarity equations $\inequality^i(\bar\var^i,\bar\var^{\neg i})\odot \bar\gamma_k^i=0$ together with the nonnegativity conditions $\bar\gamma_k^i\ge 0$. 
% Finally, $\bar\psi_k^i$ and $\bar\phi_k^i$ are the multipliers associated with the lower-level KKT equalities $\bar{\mathcal{F}}_{k+1}^i=0$: specifically, $\bar\psi_k^i$ multiplies the lower-level Lagrangian stationarity/optimality conditions, while $\bar\phi_k^i$ multiplies the lower-level complementarity equations.
% Equivalently, if $\bar{\mathcal{F}}_{k+1}^i$ is stacked as (stationarity; equalities; complementarity), then $(\bar\psi_k^i,\bar\lambda_k^i,\bar\phi_k^i)$ are the corresponding multipliers for these three blocks.

\paragraph{Innermost level $k=\numlevels^i$ for player $i$}
Consider the bottom-level problem in \cref{eqn:goop-K-level-3,eqn:goop-K-level-4}.
Let $\bar\lambda_{\numlevels^i}^i$ and $\bar\gamma_{\numlevels^i}^i$ represent the dual variables for the equality and inequality constraints, respectively. 
To simplify the notation for player $i$'s Lagrangian, we collect the dual variables $\bar\lambda_{\numlevels^i}^i$ and $\bar\gamma_{\numlevels^i}^i$ as $ \bar\eta_{\numlevels^i}^i$, i.e.,
\begin{align}
\label{eqn:complete-bottom-level-eta}
\bar\eta_{\numlevels^i}^i \coloneqq [(\bar\lambda_{\numlevels^i}^i)^\top,\ (\bar\gamma_{\numlevels^i}^i)^\top]^\top \in \mathbb{R}^{\equaldim^i+\inequaldim^i}.
\end{align}
The bottom-level Lagrangian is defined in terms of the primal variables $\bar\var^i$, the dual variables $\bar\eta_{\numlevels^i}^i$, and the other players’ primal variables $\bar\var^{\neg i}$:
\begin{align}
\bar\Lagrangian_{\numlevels^i}^i\left(\bar\var^i, \bar\var^{\neg i}, \bar\eta_{\numlevels^i}^i\right) 
    = \cost{i}_{\numlevels^i}\left( \bar\var^i, \bar\var^{\neg i}\right) - \bar\lambda_{\numlevels^i}^{i\,\top} \equality^i\left( \bar\var^i, \bar\var^{\neg i}\right) - \bar\gamma_{\numlevels^i}^{i\,\top} \inequality^i\left( \bar\var^i, \bar\var^{\neg i}\right).
\end{align}
Let $\odot$ denote the elementwise product.
The bottom-level KKT system is given by:
\begin{align}
\label{eqn:complete-kkt-bottom-level-F}
    \bar{\mathcal{F}}_{\numlevels^i}^i\left(\bar\var^i, \bar\var^{\neg i}, \bar\eta_{\numlevels^i}^i\right) 
    = 
    \begin{bNiceMatrix}[margin,cell-space-limits=1.5pt]
        \nabla_{\bar\var^i} \bar{\Lagrangian}_{\numlevels^i}^i\left(\bar\var^i, \bar \var^{\neg i}, \bar\eta_{\numlevels^i}^i\right) \\
        \equality^i( \bar\var^i, \bar\var^{\neg i})\\
        \inequality^i( \bar\var^i, \bar\var^{\neg i}) \odot \bar\gamma_{\numlevels^i}^i
    \end{bNiceMatrix} = \bm{0},
\end{align}
% and the inequality constraints
\begin{align}
\label{eqn:complete-kkt-bottom-level-G}
\bar{\mathcal{G}}^i_{\numlevels^i}\left( \bar\var^i, \bar\var^{\neg i}, \bar\gamma_{\numlevels^i}^i\right)
&= 
    \begin{bNiceMatrix}
        \inequality^i( \bar\var^i, \bar\var^{\neg i}) \\
        \bar\gamma_{\numlevels^i}^i
    \end{bNiceMatrix} \geq \bm{0}. 
\end{align}

% At any sublevel $k$, where $1\leq k\leq\numlevels^i-1$, the nested optimization problem for player $i$ can be expressed recursively as 
%     \begin{align}
%         \label{eqn:opt-problem-at-k-level}
%         \quad \min_{\var^i_k} \quad & \cost{i}_k(\var_k^i, \var^{\neg i}) \\
%     \st \quad & 
%     \begin{bNiceMatrix}[margin,cell-space-limits=1.5pt]
%     \nabla_{\var^i} \Lagrangian_{k+1}^i(\var_k^i, \var^{\neg i}, \eta_{k+1}^i) \\
%     \vdots \\
%     \nabla_{\var^i} \Lagrangian_{\numlevels^i}^i(\var_k^i, \var^{\neg i}, \eta_{\numlevels^i}^i) \\
%     \equality^i(\var_k^i, z^{\neg i})\\
%     % \inequality^i(\var^i, z^{\neg i}) \odot \gamma_{\numlevels^i:(k+1)}^i
%     \end{bNiceMatrix} = 0, 
%     % \quad  \inequality^i(\var^i, z^{\neg i}) \geq 0,
%     \end{align}
%     where the $k$-level Lagrangian function $\Lagrangian_k$ is given as \cref{eqn:k-level-Lagrangian}.
\paragraph{Upper level $k\le \numlevels^i-1$}
% At any level $k\le\numlevels^i-1$, the subproblem inherits all equality and inequality constraints from the lower levels $k+1,\dots,\numlevels^i$ 
% \david{phrasing needs to explicitly say that each level $k$ inherits both (3.9) and (3.10) as explicit constraints. the issue is that the current phrasing could be interpreted as k inherits only the constraints from lower level problems, not the entire KKT system derived from the objective and constraints of lower level problems}.

At level $k = \numlevels^i-1$, the subproblem inherits the equality constraints \cref{eqn:complete-kkt-bottom-level-F} and inequality constraints \cref{eqn:complete-kkt-bottom-level-G} from the lower level $\numlevels^i$.
Under this reformulation, dual variables $\bar\eta^i_{\numlevels^i}$ introduced at level $\numlevels^i$ reappear in the level-$(\numlevels^i-1)$ problem as additional \emph{primal} variables; we refer to these as \emph{induced primals}.
Consequently, the Lagrangian stationarity conditions are imposed not only with respect to the original primal variables $\bar\var^i$, but also with respect to the 
induced primals $\bar\eta^i_{\numlevels^i}$.
This procedure is then applied recursively at each upper level. 
In what follows, we formalize this construction of the complete KKT system for any level $k$. 

We first introduce shorthand notation for level-indexed variables used in the complete KKT system. 
At each level $k$, the complete KKT system introduces dual variables comprising: (i) $\bar\psi_k^i$, the multipliers associated with lower-level Lagrangian stationarity; (ii) $\bar\phi_k^i$, the multipliers associated with complementarity constraints; (iii) $\bar\lambda_k^i$, the multipliers for innermost-level equality constraints $\equality^i$; and (iv) $\bar\gamma_k^i$, the multipliers for innermost-level inequality constraints $\inequality^i$ and lower-level dual feasibility conditions.
We collect these dual variables into
$\bar{\eta}_k^i \coloneqq [(\bar{\psi}_k^i)^\top, (\bar{\phi}_k^i)^\top, (\bar{\lambda}_k^i)^\top, (\bar{\gamma}_k^i)^\top]^\top$,
and define
$\bar\eta_{k:\numlevels^i}^i \coloneqq [(\bar\eta_k^i)^\top, \ldots, (\bar\eta_{K^i}^i)^\top]^\top$ and 
similarly,
$\bar\gamma_{k:\numlevels^i}^i \coloneqq [(\bar\gamma_k^i)^\top, \ldots, (\bar\gamma_{K^i}^i)^\top]^\top$.
% denotes the aggregated dual variables from the lowest level $K^i$ up to level $k$.
At level $k$, $\bar{\eta}_{k+1:\numlevels^i}^i$ functions as induced primal variables.
The Lagrangian function at level $k$ is given as 
\begin{align}
\label{eqn:complete-Lagrangian-k}
    &\bar \Lagrangian_{k}^i(\bar\var^i, \bar\var^{\neg i}, \bar \eta^i_{k:\numlevels^i}) = \cost{i}_{k}(\bar\var^i, \bar\var^{\neg i}) - \bar\lambda_{k}^{i\top}h^i(\bar\var^i, \bar\var^{\neg i}) - \bar\gamma_{k,1}^{i\top}g^i(\bar\var^i, \bar\var^{\neg i}) \\ 
    &- \bar\psi_{k}^{i\top} \begin{bmatrix}
        \nabla_{\bar\var^i} \bar\Lagrangian^i_{k+1} \\[0.5ex]
        \nabla_{\bar\eta^i_{k+2:\numlevels^i}}\bar\Lagrangian^i_{k+1} \\[0.5ex]
        % \nabla_{\bar\var^i} \bar\Lagrangian^i_{k+2} \\[0.5ex]
        % \nabla_{\bar\eta^i_{k+3:\numlevels^i}}\bar\Lagrangian^i_{k+2} \\
        \vdots \\
        \nabla_{\bar\var^i} \bar\Lagrangian^i_{\numlevels^i}
    \end{bmatrix}
    - \bar\phi_{k}^{i\top} \begin{bmatrix}
        g^i(\bar\var^i, \bar\var^{\neg i}) \odot \bar\gamma^i_{k+1,1} \\[0.5ex]
        \bar\gamma^i_{k+2:\numlevels^i} \odot \bar\gamma^i_{k+1,2} \\[0.5ex]
        % g^i(\bar\var^i, \bar\var^{\neg i}) \odot \bar\gamma^i_{k+2,1} \\[0.5ex]
        % \bar\gamma^i_{k+3:\numlevels^i} \odot \bar\gamma^i_{k+2,\cdot} \\
        \vdots \\
        % g^i(\bar\var^i, \bar\var^{\neg i}) \odot \bar\gamma^i_{\numlevels^i-1,1} \\[0.5ex]
        % \bar\gamma^i_{\numlevels^i} \odot \bar\gamma^i_{\numlevels^i-1,\cdot} \\
        g^i(\bar\var^i, \bar\var^{\neg i}) \odot \bar\gamma^i_{\numlevels^i}
    \end{bmatrix} - \bar\gamma_{k,\cdot}^{i\top} \begin{bmatrix}
        \bar \gamma^i_{k+1} \\
        \vdots \\
        \bar \gamma^i_{\numlevels^i} \notag
    \end{bmatrix}. \notag
\end{align}
Using \cref{eqn:complete-Lagrangian-k}, the complete KKT system at level $k$ is 
\begin{align}
\label{eqn:k-level-complete-kkt-F}
\bar{\mathcal{F}}_{k}^i\bigl(\bar\var^{i}, \bar\var^{\neg i}, \bar \eta_{k:\numlevels^i}^i\bigr) &=  \begin{bNiceMatrix}[margin,cell-space-limits=1.5pt]
        \nabla_{\bar\var^{i}} \bar{\Lagrangian}_{k}^i\left(\bar\var^{i}, \bar \var^{\neg i}, \bar \eta_{k:\numlevels^i}^i\right) \\
        \nabla_{\bar \eta_{k+1:\numlevels^i}^i} \bar{\Lagrangian}_{k}^i\left(\bar\var^{i}, \bar \var^{\neg i}, \bar \eta_{k:\numlevels^i}^i\right) \\
        \bar{\mathcal{G}}_{k+1}^i( \bar\var^{i}, \bar\var^{\neg i}, \bar \gamma^i_{k+1:\numlevels^i}) \odot \bar\gamma_{k}^i \\
        \bar{\mathcal{F}}_{k+1}^i\left(\bar\var^{i}, \bar\var^{\neg i}, \bar \eta_{k+1:\numlevels^i}^i\right) 
    \end{bNiceMatrix} = \bm{0}, \\
\label{eqn:k-level-complete-kkt-G}
\bar{\mathcal{G}}_{k}^i\bigl(\bar\var^{i}, \bar\var^{\neg i}, \bar\gamma_{k:\numlevels^i}^i\bigr) &=  \begin{bNiceMatrix}[margin,cell-space-limits=1.5pt]
        \inequality^i( \bar\var^{i}, \bar\var^{\neg i})  \\
        \bar\gamma_{k:\numlevels^i}^i
    \end{bNiceMatrix} \geq \bm{0}. 
\end{align}
Observe that at any level $k\le\numlevels^i-1$, the subproblem inherits all equality \cref{eqn:k-level-complete-kkt-F} and inequality \cref{eqn:k-level-complete-kkt-G} constraints from the lower level $k+1$.
By this nesting property, the complete KKT system for player $i$ is achieved at the top level, $(\bar{\mathcal F_1^i},\bar{\mathcal G_1^i})$, which we denote by $(\bar{\mathcal F}^i,\bar{\mathcal G}^i)$ for convenience.
Aggregating $(\bar{\mathcal F}^i,\bar{\mathcal G}^i)$ over players yields the complete KKT system $(\bar{\mathcal F},\bar{\mathcal G})$, which provides necessary conditions for a local \ac{goop} equilibrium under \cref{assum:cost-feasible-set}.
We present the complete KKT system next. 
% \david{All proofs are deferred to the appendix?}

\begin{theorem}[Complete necessary conditions, cf. \cite{RAL-GOOP}]
\label{thm:complete-kkt-system-necessary-conditions}
% Suppose that $\bar\var^{i*}$ is a local solution to player $i$'s problem \cref{eqn:goop-K-level} and that the functions $\cost{i}_1,\dots,\cost{i}_{\numlevels^i}$, $\equality^i$ and $\inequality^i$ in \cref{eqn:goop-K-level} are continuously differentiable. 
Suppose that $\bar\var^{i*}$ is a local solution to player $i$'s GOOP problem \cref{eqn:goop-K-level} under \cref{assum:cost-feasible-set}. 
% and a constraint qualification such as the \ac{mpcc}-\ac{licq} holds at \(\bar\var^{i*}\) for each level $k \in [\numlevels^i]$.
Then, for each level $k \in [\numlevels^i]$, there exist (i) induced primal variables from lower levels $ \bar \eta_{k+1:\numlevels^i}^i$, 
and (ii) dual variables $ \bar \eta_k^i$
such that the conditions in \cref{eqn:k-level-complete-kkt-F,eqn:k-level-complete-kkt-G} are satisfied at $\big(\bar\var^{i*}, \bar\var^{\neg i*}, \bar \eta^i_{k:\numlevels^i}\big)$. 
\end{theorem}

We now show that the number of variables and conditions in the complete KKT system scales exponentially in the number of preference levels.

\begin{proposition}[Exponential growth of the complete KKT system]
\label{thm:exponential-growth-complete-kkt}
Consider player $i$'s \ac{goop} problem in \cref{eqn:goop-K-level}.
The complete KKT system $(\bar{\mathcal{F}}^i,\bar{\mathcal{G}}^i)$ has 
% \david{not sure i'd have `` = 0'' and ``>= 0'' since it makes more sense to say rows of F and G alone}
\begin{align*}
    N_{\bar\var^i,\bar\eta^i}
    &= 2^{\numlevels^i - 1}(n^i+\equaldim^i+\numlevels^i\inequaldim^i)
       &&\text{variables},\\
    N_{\bar{\mathcal{F}}^i}
    &= 2^{\numlevels^i - 1}(n^i+\equaldim^i+\numlevels^i\inequaldim^i)
       &&\text{equations in } \bar{\mathcal{F}}^i(\cdot), \ \textrm{and}\\
    N_{\bar{\mathcal{G}}^i}
    &= 2^{\numlevels^i}\inequaldim^i
       &&\text{inequalities in } \bar{\mathcal{G}}^i(\cdot).
\end{align*}
% \david{did i not already suggest renaming ``rows of'' to ``conditions in'' or similar?}
Thus, the number of variables and conditions in the complete KKT system grows exponentially with the number of preference levels.
\end{proposition}

\begin{proof}
    The proof can be found in the Appendix.
\end{proof}

The preceding result shows that the complete KKT system  $(\bar{\mathcal F},\bar{\mathcal G})$ grows linearly in the number of players but exponentially in the number of preference levels in each player's hierarchy depth. 
\Cref{sec:reduced-kkt-formulation} derives a reduced KKT system whose size grows polynomially in \emph{both} the number of players \emph{and} the number of preference levels.
% thereby substantially improving computational scalability.

\section{A Reduced Set of Necessary Conditions for GOOP Equilibria}
\label{sec:reduced-kkt-formulation}
% \begin{itemize}
%     \item Notation and Structural Setup: recursion pattern, clarify the role of $\psi_k$.
%     \item KKT System for Reduced GOOP: with equality and inequality constraints
%     \item Complexity analysis and Dimensional Growth
%     \item Structural Comparison with the Original Formulation
%     \item Note: current reduced formulation is based on QP and quadratic games 
%     \begin{itemize}
%         \item Discuss reduction in redundant variables and constraints
%     \end{itemize}
% \end{itemize}

% In this section, we propose a reduced formulation of the KKT system for the \ac{goop} problem \cref{eqn:goop-K-level}.
% Unless stated otherwise, we use the same symbols for dual variables in the reduced system as in the complete system---dropping the overbar (e.g., \(\bar\lambda \mapsto \lambda\))---to indicate that they carry out the same semantic role (e.g., equality multipliers), but are associated with the reduced system's constraints.

We derive a new formulation of the KKT system for the \ac{goop} problem \cref{eqn:goop-K-level}, which avoids the exponential scaling in \Cref{thm:complete-kkt-system-necessary-conditions}. 
Unless otherwise noted, we adopt the
notation of \Cref{sec:complete-kkt-system}, omitting overbars to denote quantities for the reduced KKT system.
% (e.g., we write $\lambda$ in place of $\bar{\lambda}$). 
% The symbols \david{what symbols? this sentence needs help} denote multipliers associated with the reduced system’s constraints but retain the same interpretive meaning.
% In the remainder of this section, we present how the reduced KKT system is constructed. 
% We begin by defining the optimality constraint $\pi_k^i = 0$, which represents the hierarchical coupling across levels: feasibility at level $k$ is constrained to decisions that are consistent with optimality of the lower levels $k+1,\dots,\numlevels^i$ \david{sentence should be streamlined}. 
We construct the reduced KKT system recursively from the innermost level.
% We begin by defining the function \(\pi_k^i(\var^i,\var^{\neg i},\eta^i_{k:\numlevels^i})=\bm{0}\), which consists of the $i^\mathrm{th}$ player's \(\klevel{k}\) stationarity condition \(\nabla_{\var^i}\Lagrangian_k^i(\var^i,\var^{\neg i},\eta^i_{k:\numlevels^i})=\bm{0}\), together with the lower-level stationarity conditions \(\nabla_{\var^i}\Lagrangian_\ell^i(\var^i,\var^{\neg i},\eta^i_{\ell:\numlevels^i})=\bm{0}\) for \(\ell=k+1,\dots,\numlevels^i\).
% We subsequently derive the reduced KKT system for player $i$ by proceeding from the innermost level (highest priority preference) to the outermost level.

% \todo{explain what $\pi(\cdot)$, $\mathcal{F} = 0$ and $\mathcal{G} \geq 0$ mean, what the dual variables denote, can we drop subscript on $\var$ and $\var^{\neg i}$?}
% \dongho{Reminder to self: 1. subsequent proofs will now proceed backwards, i.e., starting with $k = K^i$. 2. drop level subscript for $k=1$ for $\eta_1 \coloneqq \eta$? }
% \david{Turn the following into a theorem REFERS TO OLD KKT AND THAT NEW KKT IS THE RELAXED, DROPPED VERSION OF OLD KKT SYSTEM: Suppose there exists an optimal primal solution $x^*$ to \ac{goop}. Then, there exists a dual variable vector $\eta^*$ such that the following conditions hold.}
\paragraph{Innermost level $k=\numlevels^i$ for player $i$} For the bottom-level \cref{eqn:goop-K-level-3,eqn:goop-K-level-4}, we define the dual variables $\eta^i_{\numlevels^i}\coloneqq [(\lambda_{\numlevels^i}^i)^\top,\ (\gamma_{\numlevels^i}^i)^\top]^\top \in \mathbb{R}^{\equaldim^i+\inequaldim^i}.$
The corresponding Lagrangian $\Lagrangian^i_{\numlevels^i}$ and KKT system $(\mathcal F^i_{\numlevels^i},\mathcal G^i_{\numlevels^i})$ are
% \jingqi{remove $s^i$ here and mention in the algorithm - done}
% \begin{equation}
%     \eta_{\numlevels^i}^i \coloneqq \left(\lambda_{\numlevels^i}^i, \gamma_{\numlevels^i}^i\right),
% \end{equation}
% \begin{align}
%     \Lagrangian_{\numlevels^i}^i\left(\var^i, \var^{\neg i}, \eta_{\numlevels^i}^i\right) 
% &= \cost{i}_{\numlevels^i}\left(\var_{\numlevels^i}^i, z^{\neg i}\right) - \lambda_{\numlevels^i}^{i\,\top} \equality^i\left(\var_{\numlevels^i}^i, z^{\neg i}\right) - \gamma_{\numlevels^i}^{i\,\top} \inequality^i\left(\var_{\numlevels^i}^i, z^{\neg i}\right),\\
% \pi_{\numlevels^i}^i\left(\var^i, \var^{\neg i}, \eta_{\numlevels^i}^i\right)
% &= \nabla_{\var^i} \Lagrangian_{\numlevels^i}^i\left(\var^i, \var^{\neg i}, \eta_{\numlevels^i}^i\right),\\
% \label{eqn:bottom-level-reduced-kkt}
% \mathcal{F}_{\numlevels^i}^i\left(\var^i, \var^{\neg i}, \eta_{\numlevels^i}^i, s^i\right) 
% &= 
% \begin{bNiceMatrix}[margin,cell-space-limits=1.5pt]
%     \nabla_{\var^i} \Lagrangian_{\numlevels^i}^i(\eta_{\numlevels^i}^i) \\
%     \equality^i(\var_{\numlevels^i}^i, z^{\neg i})\\
%     \inequality^i(\var_{\numlevels^i}^i, z^{\neg i}) - s^i \\
%     s^i \odot \gamma_{\numlevels^i}^i
% \end{bNiceMatrix} = 0, \ 
% \mathcal{G}^i_{\numlevels^i}\left( \eta_{\numlevels^i}^i,s^i\right)
% = 
% \begin{bNiceMatrix}
%     \gamma_{\numlevels^i}^i \\
%     s^i
% \end{bNiceMatrix} \geq 0.
% \end{align}
\begin{align}
    \Lagrangian_{\numlevels^i}^i\left(\var^{i}, \var^{\neg i}, \eta_{\numlevels^i}^i\right) 
    &= \cost{i}_{\numlevels^i}\left( \var^{i}, \var^{\neg i}\right) - \lambda_{\numlevels^i}^{i\,\top} \equality^i\left( \var^{i}, \var^{\neg i}\right) - \gamma_{\numlevels^i}^{i\,\top} \inequality^i\left( \var^{i}, \var^{\neg i}\right),\\
    \pi_{\numlevels^i}^i\left(\var^{i}, \var^{\neg i}, \eta_{\numlevels^i}^i\right)
    &= \nabla_{\var^{i}} \Lagrangian_{\numlevels^i}^i\left(\var^{i}, \var^{\neg i}, \eta_{\numlevels^i}^i\right), \\
    \label{eqn:bottom-level-reduced-kkt-F}
    \mathcal{F}_{\numlevels^i}^i\left(\var^{i}, \var^{\neg i}, \eta_{\numlevels^i}^i\right) 
    &=
    \begin{bNiceMatrix}[margin,cell-space-limits=1.5pt]
        \nabla_{\var^{i}} \Lagrangian_{\numlevels^i}^i\left(\var^{i}, \var^{\neg i}, \eta_{\numlevels^i}^i\right) \\
        \equality^i( \var^{i}, \var^{\neg i})\\
        \inequality^i( \var^{i}, \var^{\neg i}) \odot \gamma_{\numlevels^i}^i
    \end{bNiceMatrix} = \bm{0}, \\
    \label{eqn:bottom-level-reduced-kkt-G}
    \mathcal{G}^i_{\numlevels^i}\left( \var^{i}, \var^{\neg i}, \gamma_{\numlevels^i}^i\right)
    &= 
    \begin{bNiceMatrix}[margin,cell-space-limits=1.5pt]
        \inequality^i( \var^{i}, \var^{\neg i}) \\
        \gamma_{\numlevels^i}^i
    \end{bNiceMatrix} \geq \bm{0}.
\end{align}
At the innermost level, the reduced KKT system \cref{eqn:bottom-level-reduced-kkt-F,eqn:bottom-level-reduced-kkt-G} coincides with the complete KKT system \cref{eqn:complete-kkt-bottom-level-F,eqn:complete-kkt-bottom-level-G}. 
% \david{worth having a remark that this is identical to the base case before?}

\paragraph{Upper level $k\le \numlevels^i-1$}
Let the function \(\pi_{k+1}^i(z^i,z^{\neg i},\eta_{k+1:K^i}^i)\) collect the stationarity conditions of levels \(k+1,\dots,K^i\). 
We associate the dual variable $\psi_k^i$ with the function $\pi_{k+1}^i$ and the dual variable $\phi_k^i$ only with the complementarity constraints involving the innermost inequality constraint $g^i$. 
% The latter construction avoids assigning dual variables for lower-level dual non-negativity constraints at every level. 
Unlike the complete KKT system, we impose stationarity of the Lagrangian function only with respect to the primal variable $\var^i$.
% not with respect to induced primals from lower levels.
Defining the dual variables $\eta^i_k \coloneqq [(\psi_k^i)^\top, ( \phi_k^i)^\top,(\lambda_k^i)^\top, (\gamma_k^i)^\top]^\top$, we write the corresponding Lagrangian $\Lagrangian^i_k$ and KKT system $(\mathcal F^i_{k},\mathcal G^i_{k})$ as
% \david{consider moving the definition of $\eta$ inline here and in (4.1), so the rest of the formatting looks prettier. also try to fix (4.10) so it is not misaligned}
\begin{align}
    % \eta_{k}^i &= \left(\lambda_{\numlevels^i:k}^i, \gamma_{\numlevels^i:k}^i, \psi_{(\numlevels^i-1):k}^i, \phi_{(\numlevels^i-1):k}^i\right), \\
\label{eqn:k-level-Lagrangian}
    \Lagrangian_{k}^i \bigl(\var^{i}, \var^{\neg i}, \eta_{k:\numlevels^i}^i\bigr) 
    &= \cost{i}_{k}\left( \var^{i}, \var^{\neg i}\right) - \lambda_{k}^{i\,\top} \equality^i\left( \var^{i}, \var^{\neg i}\right) - \gamma_{k}^{i\,\top} \inequality^i\left( \var^{i}, \var^{\neg i}\right) \\
    \phantom{=}- \psi_{k}^{i\,\top}&\pi_{k+1}^i\left(\var^{i}, \var^{\neg i},\eta_{k+1:\numlevels^i}^i\right) - \sum_{\ell=1}^{K^i-k} \phi_{k,\ell}^{i\,\top}\left[ \inequality^i\left( \var^{i}, \var^{\neg i} \right) \odot \gamma^i_{K^i-\ell+1}\right], \nonumber \\
    \pi_{k}^i\left(\var^{i}, \var^{\neg i}, \eta_{k:\numlevels^i}^i\right)
    &= \begin{bNiceMatrix}[margin,cell-space-limits=1.5pt]
        \nabla_{\var^i} \Lagrangian_{k}^i\left(\var^{i}, \var^{\neg i}, \eta_{k:\numlevels^i}^i\right) \\
        \pi_{k+1}^i\left(\var^{i}, \var^{\neg i}, \eta_{k+1:\numlevels^i}^i\right)
    \end{bNiceMatrix}, \\
\label{eqn:k-level-reduced-kkt-F}
    \mathcal{F}_{k}^i\left(\var^{i}, \var^{\neg i}, \eta_{k:\numlevels^i}^i\right)
    &= 
    \begin{bNiceMatrix}[margin,cell-space-limits=1.5pt]
        \nabla_{\var^i} \Lagrangian_{k}^i(\var^{i}, \var^{\neg i}, \eta_{k:\numlevels^i}^i) \\
         \inequality^i\left( \var^{i}, \var^{\neg i} \right) \odot \gamma_{k}^i\\
        \mathcal{F}_{k+1}^i\left(\var^{i}, \var^{\neg i}, \eta_{k+1:\numlevels^i}^i,\right)
    \end{bNiceMatrix} = \bm{0}, \\
\label{eqn:k-level-reduced-kkt-G}
    \mathcal{G}_{k}^i\left(\var^{i}, \var^{\neg i}, \gamma_{k:\numlevels^i}^i\right)
    &= 
    \begin{bNiceMatrix}[margin,cell-space-limits=1.5pt]
        \inequality^i\left( \var^{i}, \var^{\neg i} \right) \\
        \gamma_{k:\numlevels^i}^i
    \end{bNiceMatrix} \geq \bm{0}.
\end{align}
After reordering, the reduced KKT system for player $i$ $({\mathcal F}^i,{\mathcal G}^i)$ can be written as 
\begin{align}
\label{eqn:reduced-kkt-top-level}
\mathcal{F}^i\left(\var^{i}, \var^{\neg i}, \eta_{1:\numlevels^i}^i\right) 
&= 
\begin{bNiceMatrix}[margin,cell-space-limits=1.5pt]
    \nabla_{\var^i} \Lagrangian_{1}^i(\var^{i}, \var^{\neg i}, \eta_{1:\numlevels^i}^i) \\[-1.5ex]
    \vdots \\
    \nabla_{\var^{i}} \Lagrangian_{\numlevels^i}^i(\var^{i}, \var^{\neg i}, \eta_{\numlevels^i}^i) \\
    \equality^i(\var^{i}, \var^{\neg i})\\
    \inequality^i(\var^{i}, \var^{\neg i}) \odot \gamma_{1}^i \\[-1.5ex]
    \vdots \\
    \inequality^i(\var^{i}, \var^{\neg i}) \odot \gamma_{\numlevels^i}^i
\end{bNiceMatrix} = \bm{0}, \\
\label{eqn:reduced-kkt-top-level-G}
\mathcal{G}^i\left( \var^{i}, \var^{\neg i}, \gamma_{1:\numlevels^i}^i\right)
&= 
\begin{bNiceMatrix}[margin]
\inequality^i(\var^{i}, \var^{\neg i}) \\
    \gamma_{1:\numlevels^i}^i
\end{bNiceMatrix} \geq \bm{0}.
\end{align}

% \david{consider: would it be useful to color code the rows that are dropped when writing the original KKT system?}
Aggregating $({\mathcal F}^i,{\mathcal G}^i)$ over players yields the reduced KKT system $(\mathcal F,\mathcal G)$.
% The reduced KKT system can be obtained from the complete KKT system by neglecting (i) the dual variables associated with the dropped equalities from lower levels, and (ii) the stationarity equations with respect to induced primal variables. 
We characterize the relationship between the solution sets of the complete KKT system $(\bar{\mathcal F},\bar{\mathcal G})$ and reduced KKT system $(\mathcal F,\mathcal G)$ in the following result. 

\begin{theorem}[Reduced KKT system is a relaxation of the complete KKT system]
\label{thm:reduced-is-relaxation-of-complete}
Suppose \cref{assum:cost-feasible-set} holds.
Let $(\bar\var^{*}, \bar\eta^{1*}_{1:\numlevels^1},\dots, \bar\eta^{N*}_{1:\numlevels^\numplayers})$ 
% \david{why break out i and not i?} 
be a solution to 
the complete KKT system $(\bar{\mathcal{F}},\bar{\mathcal{G}})$ in \cref{eqn:k-level-complete-kkt-F,eqn:k-level-complete-kkt-G}. 
Then there exists reduced-system dual variables $\big(\eta^{i*}_{1:\numlevels^i}\big)_{i=1}^{\numplayers}$ such that $(\bar\var^{*}, \eta^{1*}_{1:\numlevels^1},\dots, \eta^{N*}_{1:\numlevels^\numplayers})$ satisfies the reduced KKT system $(\mathcal{F},\mathcal{G})$ in \cref{eqn:k-level-reduced-kkt-F,eqn:k-level-reduced-kkt-G}.
Consequently, every primal solution $\bar\var^{*}$ of the complete KKT system is a solution of the reduced KKT system and the reduced system is a relaxed set of necessary conditions that must hold at any local GOOP equilibrium.
% \david{you need to be careful about the duals not being in the same space}.
% \dongho{this is done}
\end{theorem}

\begin{proof}
    The proof can be found in the Appendix.
\end{proof}

In \cref{thm:quadratic-growth-reduced-kkt}, we prove that the number of variables and conditions in the reduced KKT system $(\mathcal F,\mathcal G)$ scales polynomially with the number of levels, in contrast to the exponential scaling of the complete KKT system. 

\begin{proposition}[Polynomial growth of reduced KKT system]
\label{thm:quadratic-growth-reduced-kkt}
Consider player $i$'s \ac{goop} problem in \cref{eqn:goop-K-level}.
The corresponding reduced KKT system $(\mathcal{F}^i, \mathcal{G}^i)$ has
\begin{align*}
    N_{\var^i,\eta^i}
    &= \Bigl(1 + \frac{\numlevels^i(\numlevels^i-1)}{2}\Bigr)n^i
        + \numlevels^i\,\equaldim^i + \frac{\numlevels^i(\numlevels^i+1)}{2}\,\inequaldim^i
       &&\text{variables},\\
    N_{\mathcal{F}^i}
    &= \numlevels^i\,n^i   + \equaldim^i + \numlevels^i\,\inequaldim^i
       &&\text{equations in } \mathcal{F}^i(\cdot), \ \textrm{and}\\
    N_{\mathcal{G}^i}
    &= (\numlevels^i+1)\,\inequaldim^i
       &&\text{inequalities in } \mathcal{G}^i(\cdot).
\end{align*}
% In particular,
% \[
%     N_{\mathrm{\var^i}}(k)
%       = \mathcal{O}\bigl(k^2 n + k(\equaldim^i + \inequaldim^i)\bigr),
%     \qquad
%     N_{\mathcal{F}}(k)
%       = \mathcal{O}\bigl(k(n+\inequaldim^i)\bigr),
%     \qquad
%     N_{\mathcal{G}}(k)
%       = \mathcal{O}\bigl(k\,\inequaldim^i\bigr),
% \]
Thus, the number of variables grows quadratically, and the number of equations and inequalities grows linearly with the number of preference levels for the reduced system.
\end{proposition}

\begin{proof}
    The proof can be found in the Appendix.
\end{proof}

The reduction in complexity naturally raises the question of whether the reduced KKT system preserves the set of primal solutions.
\Cref{thm:reduced-is-relaxation-of-complete} establishes one direction: every primal solution of the complete KKT system satisfies the reduced KKT system via reconstruction of reduced-system dual variables.
We now study the converse, namely whether a reduced-system solution can be lifted to a complete-system solution. 
We prove this converse for quadratic GOOPs. For general GOOPs with nonquadratic objectives and nonlinear constraints, establishing the same equivalence is substantially more challenging, because newly arising higher-order terms break the term-by-term correspondence between the two systems. 
\subsection{Equivalence Between Complete and Reduced KKT Systems in Quadratic GOOPs}
\label{subsec:equivalence-quadratic-goop}
% \jingqi{A Corollary for nonquadratic goop, with linear constraints. The old KKT and new KKT systems have the same primal solution. Add a theorem on equivalence in general setting for linear constraints and locally convex objective functions}
% \dongho{For now, we consider the case where all players have the same preference levels. Distinct preference levels is treated later. }
We now establish primal solution equivalence between the complete and reduced KKT systems for quadratic GOOPs. We first treat the equality-constrained case and then extend the result to the inequality-constrained case.
\subsubsection{Equality-Constrained Quadratic GOOPs}
\label{sec:equality-constained-quadratic-goop}
% Assume that each player $i$ in \cref{eqn:goop-K-level} has partially convex quadratic objective functions (in $\var^i$) and is subject only to linear equality constraints.
At any level $k$, assume that each player $i$'s objective is given by: 
\begin{equation}
\label{eq:quadratic-cost}
    \cost{i}_k(\var) \coloneqq \frac{1}{2}\var^\top Q_k^i\var + q^{i \top}_k \var,
\end{equation}
where $\var \coloneqq [(\var^i)^\top, (\var^{\neg i})^\top]^\top \in \mathbb{R}^n$, $Q_k^i \in \mathbb{R}^{n \times n}$, and $q_k^i \in \mathbb{R}^{n}$.
We partition the matrix $Q_k^i$ and vector $q^{i}_k$ into blocks $Q_k^{i,j,k} \in \mathbb{R}^{n^j \times n^k}$ and $q_k^{i,j} \in \mathbb{R}^{n^j}$ for all $i,j,k\in[\numplayers]$: 

\begin{equation}
\label{eqn:player-quadratic-cost}
Q_k^i \coloneqq
    \begin{bNiceArray}{ccc}
    Q_k^{i,1,1} & \cdots & Q_k^{i,1,\numplayers} \\
    \vdots & \ddots & \vdots \\
    Q_k^{i,\numplayers,1} & \cdots & Q_k^{i,\numplayers,\numplayers}
    \end{bNiceArray}, \quad
    q_k^i \coloneqq
    \begin{bNiceArray}{c}
    q_k^{i,1} \\
    \vdots \\
    q_k^{i,\numplayers}
    \end{bNiceArray}.
\end{equation}

We impose the following regularity assumption: 
\begin{equation}
\label{eq:quadratic-goop-regularity-assumption}
    Q_k^{i,j,j} \succeq \bm{0}, \;\; \forall i,j \in [\numplayers],
\end{equation}
and define the following matrices, which combine components in \cref{eqn:player-quadratic-cost}:
\begin{equation}
\label{eq:Q-mat-quadratic-goop}
     Q_{k}\hspace{-0.1em} :=\hspace{-0.1em} 
    \begin{bNiceArray}{cccc}
    Q_k^{1,1,1} & Q_k^{1,1,2} & \cdots &  Q_k^{1,1,N} \\[1ex]
    Q_k^{2,2,1} & Q_k^{2,2,2}& \cdots &  Q_k^{2,2,N} \\
    \vdots & \vdots & \ddots & \vdots \\
    Q_k^{\numplayers,\numplayers,1} & Q_k^{\numplayers,\numplayers,2} & \cdots & Q_k^{\numplayers,\numplayers,\numplayers}
    \end{bNiceArray}, \hspace{-0.1em}
    \hat Q_{k} \hspace{-0.1em}\coloneqq
    \begin{bNiceArray}{ccc}
    Q_k^{1,1,1} &  &  \\
     & \ddots &  \\
     &  & Q_k^{
    \numplayers,\numplayers,\numplayers}
    \end{bNiceArray},  
    q_k \coloneqq \begin{bmatrix}
        q^1_k \\ 
        \vdots\\
        q^\numplayers_k
    \end{bmatrix}.
\end{equation}
The matrix $Q_k$ consists of $i^\mathrm{th}$-block rows of $Q^i_k$.
By \cref{eq:quadratic-goop-regularity-assumption}, the block-diagonal matrix $\hat{Q}_k$ is positive semidefinite, and, thus, player $i$'s objective is partially convex in $\var^i$, i.e., $\nabla_{\var^i}^2 J^i_k $ is a positive semidefinite matrix for each player $i\in [N]$ at each level $k\in [K]$.
 
The linear equality constraints imposed on each player $i$ take the form of: 
\begin{equation}
\label{eq:linear-equality-constraints}
    H^i\var = {h}^i, \quad H^i \coloneqq
    \begin{bmatrix}
    H^i_{1} &\cdots & H^i_{\numplayers}
    \end{bmatrix}, 
\end{equation}
for matrix $H^i \in \mathbb{R}^{\equaldim^i \times n}$, vector $h^i \in \mathbb{R}^{\equaldim^i}$, and submatrices $H_j^i \in \mathbb{R}^{\equaldim^i \times n^j}$.
Define
\begin{equation}
    \label{eq:H-br-H-bd-quadratic-goop}
    H \coloneqq
    \begin{bNiceArray}{ccc}
    H^1_{1} &\cdots & H^1_{\numplayers} \\
    \vdots &\ddots & \vdots \\
    H^\numplayers_{1} & \cdots & H^\numplayers_{\numplayers} \\
    \end{bNiceArray}, \quad 
    \hat H \coloneqq
    \begin{bNiceArray}{ccc}
    H^1_{1} &  & \\
      &\ddots &   \\
      &  & H^\numplayers_{\numplayers} \\
    \end{bNiceArray}, \quad h\coloneqq \begin{bmatrix}
        h^1\\ \vdots \\ h^N
    \end{bmatrix}.
\end{equation}
The matrix $H$ is the block-row concatenation of $\{H^i\}_{i=1}^{N}$, and the matrix $\hat H$ is a block-diagonal matrix that retains only the diagonal blocks $H_{i}^i$ from $H$. 
\begin{assumption}[Full row rank of $\hat H$]
\label{ass:hatH-full-row-rank}
We assume that the block-diagonal matrix \(\hat H\) has full row rank. Equivalently, each diagonal block \(H_i^i\) has full row rank.
\end{assumption}

\paragraph{Different numbers of preference levels}
Players need not share a common number of preference levels. 
For analytical convenience, we regard each player as having levels $k\in\{1,\dots,\numlevels\}$ where $\numlevels \coloneqq \max\{K^i : i\in[\numplayers]\}.$
If player $i$ has no level $k$ (i.e., $\numlevels^i < k \leq \numlevels$), we set the player's objective to $ Q_k^i = \bm{0},\; q_k^i = \bm{0}$,
% so that this level constitutes no preference. 
which retains the regularity assumption \cref{eq:quadratic-goop-regularity-assumption} and positive semidefiniteness of $\hat Q_k$. 
% In what follows, we use the common index range $k=1,\dots,\numlevels$.

% \todo{add dimensions?, appendix an example of 2-player bilevel, trilevel quadratic goop}
Using the objective and constraint terms defined above, we develop the recursive structure of the complete and reduced KKT systems. 
% The key idea in the recursion is that at any given level $k$, the coefficient matrix $\bar M_k$ \david{unclear what this is the coefficient of, and which KKT system it refers to} (respectively, $M_k$) contains the coefficient matrix of the immediate one lower level $\bar M_{k+1}$ (respectively, $M_{k+1}$).  
The key idea in the recursion is that, at each level $k$, the complete KKT system matrix $\bar M_k$ (resp., the reduced KKT system matrix $M_k$) is \emph{nested}: it contains, as a principal block, the level-$(k+1)$ KKT system matrix $\bar M_{k+1}$ (resp., $M_{k+1}$).
This enforces the inter-level optimality chain where the upper-level's decision space is constrained by lower-level optimality. 
Furthermore, we will repeatedly use the following matrix for the recursive structure:
\begin{equation}
\label{eq:bar-R-k-quadratic-goop}
    \bar R_{k} \coloneqq
\begin{bNiceArray}[margin]{cc|c}
\hat Q_k & \bm{0} & \Block{2-1}{\bar R_{k+1}}  \\
       \bm{0} & \bm{0} &                   \\
\hline
\Block{2-2}{\bar R_{k+1}} &  & \Block{2-1}{\bm{0}}   \\
&  &  
\end{bNiceArray}, \quad \forall k \in [\numlevels-1], \quad \bar R_\numlevels \coloneqq \begin{bmatrix}
    \hat Q_{\numlevels} & \hat{H}^\top \\
    \hat H & \bm{0}
\end{bmatrix}.
\end{equation}
By construction, $\bar R_k$ is symmetric for all $k$.
% By \cref{thm:quadratic-growth-reduced-kkt}, the reduced KKT system has fewer variables and constraints than the complete KKT system.
To align dimensions between the complete and reduced systems, 
% we pad the reduced KKT system with zero rows and treat the corresponding variables as placeholders (see \cref{eq:reduced-kkt-quadratic-goop}).
we use the matrix
% define
\begin{equation}
\label{eqn:R_{k+1}}
    R_{k+1} \coloneqq \begin{bNiceArray}[margin]{c}
    \bar R_{k+1,[1:n,:]} \\[0.5ex]
    \hline
    \bm{0}
\end{bNiceArray},
\end{equation} which takes the first $n$ rows of the matrix $\bar R_{k+1}$ and appends zeros below (see \cref{eq:reduced-kkt-quadratic-goop}). 
These $n$ rows correspond to the stationarity condition with respect to the primal variable $z$ in \cref{eqn:k-level-reduced-kkt-F}.
Moreover, we augment playerwise lower-level dual variables (e.g., $\bar\lambda_\numlevels = [(\bar\lambda^1_\numlevels)^\top,\dots,(\bar\lambda^\numplayers_\numlevels)^\top]^\top$) and use the notation $\bar\eta_{\numlevels:k+1} = [(\bar\eta_\numlevels)^\top,\dots,(\bar\eta_{k+1})^\top]^\top$. 
We present the resulting complete and reduced KKT systems next.

\textbf{Recursive structure of the complete KKT system.}
The conditions \cref{eqn:k-level-complete-kkt-F} for a quadratic \ac{goop} at any level $k$ admit the following linear representation: 
% \david{bit awkward formatting with the for all k qualifier and the commas}
\begin{equation}
    \label{eq:complete-kkt-quadratic-goop}
    \begin{bNiceArray}[margin]{cc|c}
    Q_k & \bm{0} & \Block{2-1}{\bar R_{k+1}}  \\
           \bm{0} & \bm{0} &                   \\
    \hline
    \Block{2-2}{\bar M_{k+1}} &  & \Block{2-1}{\bm{0}}   \\
    &  &    
    \CodeAfter
    %     \OverBrace[ , yshift=3pt]{1-1}{1-1}{c_k}
    \UnderBrace[ , yshift=4pt]{4-1}{4-3}{\bar M_k}
    \end{bNiceArray} 
        \begin{bNiceArray}[margin]{c}
        \bar \var \\
        % \begin{bmatrix}\bar \lambda_\numlevels \\
        % \bar \psi_{\numlevels-1} \\[-1ex]
        % \vdots \\[-1ex]
        % \bar \lambda_{k+1}\end{bmatrix} \\
        \bar \eta_{\numlevels:k+1} \\[1ex]
        \hline\\[-2ex]
        \bar \eta_{k}
        % \begin{bmatrix}\bar \psi_{k} \\
        % \bar \lambda_{k}\end{bmatrix}
        % \bar \psi_{k,1}\\
        % \bar \lambda_{k,1}\\
        % \vdots
        % \bar \psi_{k,2^{\numlevels-k}}\\
        % \bar \lambda_{k,2^{\numlevels-k}}
    \CodeAfter
    \UnderBrace[ , yshift=10pt]{3-1}{3-1}{\bar v_k}
    \end{bNiceArray} =         
    \begin{bNiceArray}[margin]{c}
        q_k \\
        \bm{0} \\
        \hline \\[-2ex]
        q_{k+1}  
    \CodeAfter
    \UnderBrace[ , yshift=4pt]{4-1}{4-1}{\bar p_k}
    \end{bNiceArray}, \hspace{0.4em} \forall k\in[\numlevels-1].\vspace{1.2em}
% \bigskip
\end{equation}

\textbf{Recursive structure of the reduced KKT system.}
The conditions \cref{eqn:k-level-reduced-kkt-F} for a quadratic \ac{goop} at any level $k$ can be expressed as:
\begin{equation}
    \label{eq:reduced-kkt-quadratic-goop}
    \begin{bNiceArray}[margin]{cc|c}
    Q_k & \bm{0} & \Block{2-1}{R_{k+1}}  \\
           \bm{0} & \bm{0} &                   \\
    \hline
    \Block{2-2}{ M_{k+1}} &  & \Block{2-1}{\bm{0}}   \\
    &  &    
    \CodeAfter
    %     \OverBrace[ , yshift=3pt]{1-1}{1-1}{c_k}
    \UnderBrace[ , yshift=4pt]{4-1}{4-3}{M_k}
    \end{bNiceArray} 
        \begin{bNiceArray}[margin]{c}
        \var \\
        % \begin{bmatrix}\bar \lambda_\numlevels \\
        % \bar \psi_{\numlevels-1} \\[-1ex]
        % \vdots \\[-1ex]
        % \bar \lambda_{k+1}\end{bmatrix} \\
        \eta_{\numlevels:k+1} \\[1ex]
        \hline\\[-2ex]
        \eta_k
        % \begin{bmatrix}\bar \psi_{k} \\
        % \bar \lambda_{k}\end{bmatrix}
        % \bar \psi_{k,1}\\
        % \bar \lambda_{k,1}\\
        % \vdots
        % \bar \psi_{k,2^{\numlevels-k}}\\
        % \bar \lambda_{k,2^{\numlevels-k}}
    \CodeAfter
    \UnderBrace[ , yshift=10pt]{3-1}{3-1}{v_k}
    \end{bNiceArray} =         
    \begin{bNiceArray}[margin]{c}
        q_k \\
        \bm{0} \\
        \hline \\[-2ex]
         q_{k+1}  
    \CodeAfter
    \UnderBrace[ , yshift=4pt]{4-1}{4-1}{p_k}
    \end{bNiceArray}, \hspace{0.4em} \forall k\in[\numlevels-1].\vspace{1.20em}
% \bigskip
\end{equation}

For the base level $k=\numlevels$, the complete and reduced KKT systems coincide: 
% i.e., $  M_\numlevels \coloneqq 
%     \bar M_\numlevels, v_\numlevels \coloneqq \bar v_\numlevels$ and $p_\numlevels \coloneqq \bar p_\numlevels$.
\begin{equation}
        \underbrace{\begin{bmatrix}
    Q_{\numlevels} & \hat{H}^\top \\
    H & \bm{0}
\end{bmatrix}}_{M_\numlevels = \bar M_\numlevels} \underbrace{\begin{bmatrix}
     z \\
    \lambda_\numlevels
\end{bmatrix}}_{v_\numlevels = \bar v_\numlevels}\hspace{-0.1em} =\hspace{-0.1em} \underbrace{\begin{bmatrix}
    q_\numlevels \\
    h
\end{bmatrix}}_{p_\numlevels = \bar p_\numlevels}.
\end{equation}
For any $k\in[\numlevels]$, the top-left blocks of $\bar M_k$ in \cref{eq:complete-kkt-quadratic-goop} and $M_k$ in \cref{eq:reduced-kkt-quadratic-goop} contain the same $ Q_k$ term on the leading diagonal; thus, the two matrices share the first $n$ columns.
Furthermore, because the right-hand-side vector is identical, $p_k = \bar p_k$, it follows that $\bar M_k \bar v_k = M_k v_k$. 
This fact is the basis for the next result, where we establish equivalence of the primal solution sets via a column-space relationship between $\bar M_k$ and $M_k$.

\begin{theorem}[Primal solution equivalence in quadratic \acp{goop}]
\label{thm:ng-og-primal-equivalence}
Consider a GOOP problem \cref{eqn:goop-K-level} with partially convex quadratic objectives \cref{eq:quadratic-cost} and linear equality constraints \cref{eq:linear-equality-constraints} under \cref{ass:hatH-full-row-rank}.
% Let $k$ denote the preference level. 
Then the following hold:
\begin{itemize}
  \item[(a)] 
    $\mathrm{Col}\big(R_{k}\big) \subseteq \mathrm{Col}\big(\bar R_{k}\big), ~\forall k \in [\numlevels].$
  \item[(b)] Denote the first $n$ columns shared by both matrices $\bar M_k$ and $M_k$ as $C_k$, 
  and consider the partition: $\bar M_k = \big[ C_k ~|~ \bar M_k^c\big]$ and $M_k = \big[ C_k ~|~ M_k^c\big]$. 
Then we have 
$\mathrm{Col}\big(M_{k}^c\big) \subseteq \mathrm{Col}\big(\bar M_{k}^c\big), ~\forall k\in[\numlevels-1]$. 
  \item[(c)] The primal solution sets of the complete KKT system \cref{eq:complete-kkt-quadratic-goop} and reduced KKT system \cref{eq:reduced-kkt-quadratic-goop} coincide for all $k\in[\numlevels]$.
\end{itemize}
\end{theorem}

\begin{proof}
    The proof can be found in the Appendix.
\end{proof}

\cref{thm:ng-og-primal-equivalence} formalizes the structural relationships that ensure coincidence of the primal solution sets (even though the corresponding dual variables may differ).

\subsubsection{Inequality-Constrained Quadratic GOOP}
\label{sec:inequality-constained-quadratic-goop} 
We now extend the preceding result in \Cref{sec:equality-constained-quadratic-goop} to quadratic GOOPs with linear inequality constraints.
Consider a \ac{goop} problem \cref{eqn:goop-K-level} with quadratic objectives \cref{eq:quadratic-cost} and linear equality constraints \cref{eq:linear-equality-constraints}.
Assume that each player $i$ also has inequality constraints: 
\begin{equation}
    \label{eq:linear-inequality-constraints}
    G^i \var \geq g^i, \quad
    G^i \coloneqq \begin{bmatrix}
    G^i_{1} &\cdots & G^i_{\numplayers}
    \end{bmatrix},
\end{equation}
for matrix $G^i \in \mathbb{R}^{\inequaldim^i \times n}$, vector $ g^i \in \mathbb{R}^{\inequaldim^i}$ and submatrices $G^i_j \in \mathbb{R}^{\inequaldim^j \times n^j}$.
% The inequality-constrained \ac{goop} takes the form of 
% \begin{equation}
%     \label{eqn:goop-K-level-inequality}
% \begin{aligned}
%     \min_{\var^i_1} \quad & \frac{1}{2}\var_1^{\top} Q^i_1 \var_1 + r_1^{\top} \var_1\\
%     \st \quad & \var_1 \in \argmin_{\var_2} \: \frac{1}{2}\var_2^{\top} Q^i_2 \var_1^2 + r_2^{\top} \var_2 \\
%     & \phantom{\st\quad\qquad} \rotatebox[origin=c]{150}{$\ddots$}\\[-1em]
%     & \phantom{\st\qquad\,} \st \quad \var_{K-1} \in \argmin_{\var_K} \: \frac{1}{2}\var_K^{\top} Q^i_K \var_K + r_K^{\top} \var_K \\
%     & \qquad\phantom{\st\qquad\st\qquad\qquad} \st \quad H^i \var_K - \tilde h^i = 0, \; G^i \var_K - \tilde g^i \geq 0.
% \end{aligned}
% \end{equation}
Define the active set at a solution $\var^*$ as
\begin{equation}
    \mathcal A^i(\var^*) \coloneqq \{\, j \mid (G^i \var^* - g^i)_j = 0\,\}. 
\end{equation}
Let matrix $G^i_{\mathcal{A}^i}$ be the submatrix of $G^i$ with rows indexed by $\mathcal A^i(\var^*)$, and likewise index the entries of the vector $g^i$ corresponding to active inequality constraints as $g^i_{\mathcal{A}^i}$. 
Define the vector $\tilde h^i \coloneqq [h^i;g^i_{\mathcal{A}^i}]$.
We present the primal solution equivalence result for the inequality-constrained case in \cref{thm:ng-og-primal-equivalence-inequality}.
\begin{theorem}[Primal solution equivalence in inequality-constrained quadratic \acp{goop}]
\label{thm:ng-og-primal-equivalence-inequality}
Consider a \ac{goop} problem \cref{eqn:goop-K-level} with partially convex quadratic objectives \cref{eq:quadratic-cost}, linear equality constraints \cref{eq:linear-equality-constraints}, and inequality constraints \cref{eq:linear-inequality-constraints} for each player.
Assume that at every solution $\var^{*}$, the following hold for player $i$:

\textbf{1. Strict complementarity.}
    At the innermost level $\numlevels^i$, 
    % where $\myhl{\numlevels \coloneqq \max\{K^i : i\in[\numplayers]\}}$, 
    every active inequality constraint has a positive multiplier, i.e., 
    $ (G^i\var^{*}-g^i)_j = 0$ and
    $\gamma^{i*}_{\numlevels^i,j} > 0, \;\; \forall j \in \mathcal{A}^i(z^{*})$.

\textbf{2. Regularity.} The matrix  
    \begin{equation}
    \label{eq:inequality-licq}
         \tilde H \coloneqq \begin{bNiceArray}{ccc}
    \begin{bmatrix}
        H^1_1 \\
        G^1_{1,\mathcal{A}^1}
    \end{bmatrix} &  &   \\
      &\ddots &   \\
      &   & \begin{bmatrix}
        H^\numplayers_\numplayers \\
        G^\numplayers_{\numplayers,\mathcal{A}^\numplayers}
    \end{bmatrix} \\
    \end{bNiceArray}
    \end{equation} has full row rank. Equivalently, each diagonal block $\begin{bmatrix}H_i^i\\G_{i,\mathcal{A}^i}^i\end{bmatrix}$ has full row rank.
    
Then, the set of primal solutions of the reduced KKT system $({\mathcal{F}},{\mathcal{G}})$ coincides with that of the complete KKT system $(\bar{\mathcal{F}},\bar{\mathcal{G}})$. 
% \david{I agree with Jingqi's margin comment.}
\end{theorem}

\begin{proof}
    The proof can be found in the Appendix.
\end{proof}

The strict complementarity assumption in \cref{thm:ng-og-primal-equivalence-inequality} provides the nondegeneracy needed at the innermost level to reconstruct the relevant reduced-system dual variables and initiate the recursive argument via \cref{thm:ng-og-primal-equivalence} across higher levels.
\begin{remark}
    The assumption of strict complementarity at the innermost level guarantees the local reduction to an equality-constrained subproblem. Characterizing primal equivalence under weaker conditions (even in quadratic GOOPs) remains a challenge as the active set may no longer be uniquely identified by the strictly positive multipliers in these degenerate cases.
\end{remark}

\subsection{From Quadratic GOOPs to Nonquadratic GOOPs}
\label{sec:nonquadratic-goop}

The preceding analysis in
\Cref{sec:equality-constained-quadratic-goop,sec:inequality-constained-quadratic-goop} raises the question of whether the reduced and complete KKT systems continue to share the same primal solution set for general smooth GOOPs with nonquadratic objectives and nonlinear constraints.
Empirically, we observe such primal equivalence in \Cref{sec:experiment}, but a general proof is nontrivial.

The difficulty is structural: when players' objectives are nonquadratic, higher-order derivatives of inner-level objectives introduce additional terms that destroy the block-wise correspondence between the two KKT systems.
Consequently, lower-level dual variables---which generally take different values in the complete and reduced systems---generate \emph{different coefficients} in the upper-level nonlinear stationarity conditions.
This misalignment makes it challenging to reconstruct dual variables for the complete KKT system from a reduced KKT-system primal solution.

The coefficient mismatch \emph{first} appears at level $k = \numlevels^i-2$ when player $i$ is subject only to linear equality constraints \cref{eq:linear-equality-constraints}. 
The multipliers for the complete system are
$\bar\eta^i_{\numlevels^i-2}\coloneqq[(\bar\psi^i_{\numlevels^i-2})^\top,(\bar\lambda^i_{\numlevels^i-2})^\top]^\top$.
The complete-system Lagrangian stationarity is
\begin{align}
\label{eq:nonquadratic-goop-stationarity-condition}
\nabla_{\bar\var^i}\bar\Lagrangian^i_{\numlevels^i-2}(\bar\var,\bar\eta^i_{\numlevels^i-1:\numlevels^i},\bar\eta^i_{\numlevels^i-2})
=&\;
\nabla_{\bar\var^i}\cost{i}_{\numlevels^i-2}(\bar\var)
-\Bigl(\nabla_{\bar\var^i}^2\bar\Lagrangian^i_{\numlevels^i-1}(\bar\var,\bar\eta^i_{\numlevels^i-1:\numlevels^i})\Bigr)^\top \bar\psi^i_{\numlevels^i-2,1}
\notag\\
&\;
-\Bigl(\nabla_{\bar\var^i}^2\bar\Lagrangian^i_{\numlevels^i}(\bar\var,\bar\eta^i_{\numlevels^i})\Bigr)^\top \bar\psi^i_{\numlevels^i-2,2}
- H_i^{i\top}\bar\lambda^i_{\numlevels^i-2}.
\end{align}
The stationarity equation for the reduced system is obtained analogously.
Next, expand the Hessian of the lower level ($\numlevels^i-1$) Lagrangian appearing in \cref{eq:nonquadratic-goop-stationarity-condition}:
\begin{equation}
\label{eq:nonquadratic-goop-stationarity-condition-2}
\nabla_{\bar\var^i}^2\bar\Lagrangian^i_{\numlevels^i-1}(\bar\var,\bar\eta^i_{\numlevels^i-1:\numlevels^i})
=
\nabla_{\bar\var^i}^2\cost{i}_{\numlevels^i-1}(\bar\var)
-
\Bigl(\nabla_{\bar\var^i}^3\cost{i}_{\numlevels^i}(\bar\var)\Bigr)^\top \bar\psi^i_{\numlevels^i-1}.
\end{equation}
The final term in \cref{eq:nonquadratic-goop-stationarity-condition-2} arises from the third derivative of the innermost objective and is weighted by the lower-level dual variable $\bar\psi^i_{\numlevels^i-1}$.
Notice that $\bar\psi^i_{\numlevels^i-1}$ need not coincide with its reduced \acs{kkt} system counterpart $\psi^i_{\numlevels^i-1}$.
As a result, the misaligned dual variables generate different coefficients in the upper-level stationarity condition \cref{eq:nonquadratic-goop-stationarity-condition}, i.e., $\nabla_{\bar\var^i}^2\bar\Lagrangian^i_{\numlevels^i-1}(\bar\var,\bar\eta^i_{\numlevels^i-1:\numlevels^i}) \neq \nabla_{\var^i}^2\Lagrangian^i_{\numlevels^i-1}(\var,\eta^i_{\numlevels^i-1:\numlevels^i})$.  
In this way, the complete and reduced systems progressively accumulate mismatched coefficients in their stationarity conditions. The proof of 
\cref{thm:ng-og-primal-equivalence} relied upon consistent coefficients to align the dual variables between the two KKT systems. 
Consequently, establishing primal solution equivalence in this more general regime requires a different mechanism to align---or otherwise relate---the relevant dual variables across the two systems.

Even in the absence of a primal-equivalence guarantee in general settings, we can still certify the local optimality of candidate solutions for general GOOPs. 
Accordingly, we shift our focus to establishing second-order sufficient conditions (SOSC) that characterize when a candidate solution corresponds to a local equilibrium. 
\section{Second-order Sufficient Conditions}

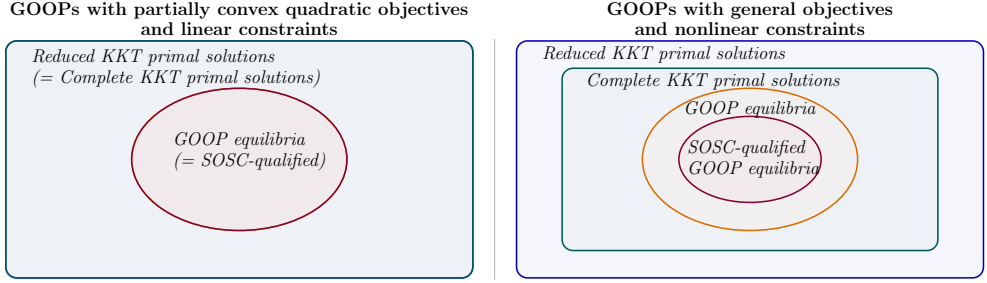
\begin{figure}[t]
\centering
\resizebox{1.0\linewidth}{!}{%
\begin{tikzpicture}[font=\Large]

\tikzset{
  setrect/.style={
    draw=#1, very thick,
    rounded corners=3mm,
    fill=#1!12, fill opacity=0.32
  },
  setell/.style={
    draw=#1, very thick,
    fill=#1!12, fill opacity=0.32
  },
  setlabel/.style={align=left, anchor=west}
}

% =========================================================
% Left panel: Quadratic GOOP (equalities)
% =========================================================
\begin{scope}[shift={(-7.0,0)}]
\coordinate (Oq) at (0,0);

% Slightly larger outer sets
\node[setrect=blue!70!black, minimum width=12.8cm, minimum height=6.5cm] (Rq) at (Oq) {};
\node[setrect=teal!70!black, minimum width=12.8cm, minimum height=6.5cm] (Fq) at (Oq) {};

% Larger coincident ellipses
\path[setell=orange!85!black]
  (Oq) ellipse [x radius=2.95cm, y radius=1.95cm];
\path[setell=purple!70!black]
  (Oq) ellipse [x radius=2.95cm, y radius=1.95cm];

% Labels
\node[setlabel] at ($ (Rq.north west) + (0.60,-0.78) $) {%
\emph{Reduced KKT primal solutions}\\[-0.1em]
\emph{(= Complete KKT primal solutions)}
};

\node[setlabel] at ($ (Oq) + (-1.95,0.22) $) {%
\emph{GOOP equilibria}\\[-0.1em]
\emph{(= SOSC-qualified)}
};

% Panel title
\node[align=center] at ($ (Rq.north) + (0,0.60) $)
{\textbf{\shortstack{GOOPs with partially convex quadratic objectives\\and linear constraints}}};
\end{scope}

% =========================================================
% Right panel: General GOOP (inclusions)
% =========================================================
\begin{scope}[shift={(7.0,0)}]
\coordinate (Og) at (0,0);

% Outer rectangles
\node[setrect=blue!70!black, minimum width=12.8cm, minimum height=6.5cm] (Rg) at (Og) {};
\node[setrect=teal!70!black, minimum width=10.3cm, minimum height=5.0cm] (Fg) at (Og) {};

% Larger inner ellipses
\path[setell=orange!85!black]
  (Og) ellipse [x radius=2.95cm, y radius=1.95cm];
\path[setell=purple!70!black]
  (Og) ellipse [x radius=1.95cm, y radius=1.18cm];

% Labels
\node[setlabel] at ($ (Rg.north west) + (0.60,-0.40) $) {%
\emph{Reduced KKT primal solutions}
};

\node[setlabel] at ($ (Fg.north west) + (0.55,-0.38) $) {%
\emph{Complete KKT primal solutions}
};

\node[setlabel] at ($ (Og) + (-1.95,1.38) $) {%
\emph{GOOP equilibria}
};

\node[setlabel] at ($ (Og) + (-1.85,0.02) $) {%
\emph{SOSC-qualified}\\[-0.1em]
\emph{GOOP equilibria}
};

% Panel title
\node[align=center] at ($ (Rg.north) + (0,0.60) $)
{\textbf{\shortstack{GOOPs with general objectives \\and nonlinear constraints}}};
\end{scope}

% Separator
\draw[black!30, line width=0.5pt] (0,-3.30) -- (0,3.30);

\end{tikzpicture}%
}\vspace{-0.4em}
\caption{
(\textbf{Left}) For quadratic objectives and linear constraints, under the assumptions of \Cref{subsec:equivalence-quadratic-goop}, the reduced and complete KKT systems share the same primal solution set, and every GOOP equilibrium is SOSC-qualified (hence the corresponding sets coincide).
(\textbf{Right}) For general smooth GOOPs, the reduced KKT primal solution set is a relaxation (superset) of the complete KKT primal solution set, which contains the GOOP equilibria; \cref{thm:sosc-goop} certifies a subset of such equilibria. 
\label{fig:set-inclusions-rect-ell}}
\end{figure}

In this section, we describe sufficient conditions for local optimality in a general GOOP problem. 
As illustrated in \cref{fig:set-inclusions-rect-ell}, a solution to the reduced KKT system may be \emph{spurious}: it satisfies the reduced system, yet may fail to satisfy the complete KKT system and therefore cannot correspond to a GOOP equilibrium. 
We present \cref{thm:sosc-goop}, which provides a set of conditions that each subproblem (for every player and level) must satisfy in order for a candidate solution to the reduced KKT system to be a true local GOOP equilibrium.

We proceed along the lines established by \cite{fiacco1968}, in the context of identifying weak local minimizers.
Using the reduced KKT system
\cref{eqn:k-level-reduced-kkt-F,eqn:k-level-reduced-kkt-G}, we define the linearized feasible cone and the critical cone of the
$\klevel{k}$ subproblem.
First, we define player $i$'s $\klevel{k}$ active set at $(\var^*, \eta^{i*}_{k+1:\numlevels^i})$
\begin{equation}
\mathcal{A}^i_k(\var^*,\eta^{i*}_{k+1:\numlevels^i}) \coloneqq \{\, j \;:\; \mathcal{G}^i_{k,j}(\var^*, \eta^{i*}_{k+1:\numlevels^i}) = 0 \,\},
\end{equation}
where $\mathcal{G}^i_{k}$ refers to player $i$'s reduced KKT system (inequalities) in \cref{eqn:k-level-reduced-kkt-G}.
% \david{the notation in the subcript of $\mathcal{G}_{k, j}$ is not adequately introduced}
The set of linearized feasible directions is
\begin{equation}
\label{eq:lin-feas-cone}
\mathcal{L}_F\bigl(\var^*,\eta^{i*}_{k+1:\numlevels^i}\bigr)
=
\left\{
d \ \middle|\ 
\begin{aligned}
&\nabla_{(\var^i,\eta^i_{k+1:\numlevels^i})}\mathcal{F}_k^i\bigl(\var^*,\eta^{i*}_{k+1:\numlevels^i}\bigr)^\top d = \bm{0},\\
&\nabla_{(\var^i,\eta^i_{k+1:\numlevels^i})}\mathcal{G}^{i}_{k,j}\bigl(\var^*,\eta^{i*}_{k+1:\numlevels^i}\bigr)^\top d \ge \bm{0},
\; \forall j \in \mathcal{A}^i_k %(\var^*,\eta^{i*}_{k+1:\numlevels^i})
\end{aligned}
\right\}.
\end{equation}
The critical cone is then
\begin{equation}
\label{eq:critical-cone}
\mathcal{C}_k\bigl(\var^*,\eta^{i*}_{k+1:\numlevels^i}\bigr)
=
\left\{
d \in \mathcal{L}_F
\ \middle|\
\nabla_{(\var^i,\eta^i_{k+1:\numlevels^i})}\mathcal{G}^i_{k,j}\bigl(\var^*, \eta^{i*}_{k+1:\numlevels^i}\bigr)^\top d = \bm{0},
\ \forall j \in \mathcal{A}^{i+}_k
\right\},
\end{equation}
where $\mathcal{A}^{i+}_k \coloneqq \{\, j \in \mathcal{A}^i_k(\var^*,\eta^{i*}_{k+1:\numlevels^i}) : \gamma^{i*}_{k,j} > 0 \,\}$,
and the feasible neighborhood set is 
\begin{equation}
\label{eqn:goop-sosc-Pk}\mathcal{P}_k(\epsilon, \delta) := \left\{ p \;\middle|\; 
    \begin{aligned}
        & \| p - d \| \leq \epsilon \text{ for some critical direction } d \in \mathcal{C}_k(\var^*, \eta^{i*}_{k+1:\numlevels^i}), \\
        &(\var^{i*}, \eta^{i*}_{k+1:\numlevels^i}) + \tilde\delta p \text{ satisfies } \cref{eqn:k-level-reduced-kkt-F} \text{ and } \cref{eqn:k-level-reduced-kkt-G} \text{ for some }  \tilde\delta \text{ such that } \\
        & \qquad 0 < \tilde\delta < \delta, \text{ and } \| p \| = 1.
    \end{aligned}
\right\}.
\end{equation}

We now present a second-order sufficient condition for the $\klevel{k}$ subproblem of player $i$ using the reduced KKT system. 
\begin{theorem}[Second-order sufficient condition for local GOOP equilibrium]
\label{thm:sosc-goop}
Let $(\var^{i*}, \var^{\neg i*}, \eta_{1:\numlevels^i}^{i*})$ be a solution to the reduced KKT system in \cref{eqn:k-level-reduced-kkt-F,eqn:k-level-reduced-kkt-G} under \cref{assum:cost-feasible-set}. 
Suppose that, for each level $k \in [\numlevels^i-1]$, the $\klevel{k}$ Lagrangian function $\Lagrangian^i_k$ is stationary with respect to $\eta_{k+1:\numlevels^i}^i$, i.e., 
$\nabla_{\eta_{k+1:\numlevels^i}^i}\Lagrangian^i_k(\var^{i*}, \var^{\neg i*}, \eta_{k+1:\numlevels^i}^{i*}) = 0$. 
Suppose that there exists $\epsilon' > 0$, $\delta' > 0$ such that for every $p \in \mathcal{P}_k(\epsilon', \delta')$, we have 
\begin{equation}
\label{eq:local-psd-primal}
    p^\top \nabla^2_{\left(\var^i, \eta_{k+1:\numlevels^i}^i\right)} \Lagrangian_k^i \bigl( (\var^{i*}, \eta_{k+1:\numlevels^i}^{i*}) + \alpha\tilde{\delta} p, \eta_{k}^{i*} \bigr) p \geq 0
\end{equation}
for all $\alpha \in (0, 1)$.
Then, $\var^{i*}$ is a weak (possibly non-isolated) local minimizer for player $i$'s $\klevel{{k}}$ subproblem.
Furthermore, if \cref{eq:local-psd-primal} holds at every level \(k\in[\numlevels^i-1]\) for every player \(i\in[\numplayers]\), then the solution of the reduced KKT system $\left(\var^*,\eta^{1*}_{1:\numlevels^1},\dots, \eta^{\numplayers*}_{1:\numlevels^\numplayers}\right)$ is a weak local solution of the GOOP problem and satisfies \cref{def:goop-gnep}.

\end{theorem}

\begin{proof}
% By \cite[Theorem~2.1]{fiacco1968}, $\var^{i*}$ is a weak (non-isolated) local minimizer of player $i$'s $\klevel{k}$ subproblem.
% Apply the same argument recursively for every player $i\in[\numplayers]$ and each level
% $k\in[\numlevels^i-1]$. If \cref{eq:local-psd-primal} holds at all such subproblems, then each player’s
% lower-level responses are locally minimizing at every level, and therefore $\left(\var^*,\eta^{1*}_{1:\numlevels^1},\dots, \eta^{\numplayers*}_{1:\numlevels^\numplayers}\right)$ constitutes a weak local minimizer of the GOOP problem.
By \cite[Theorem~2.1]{fiacco1968}, condition \cref{eq:local-psd-primal} implies that $(\var^{i*}, \eta^{i*}_{k+1:\numlevels^i})$ is a weak local minimizer of player $i$'s \(\klevel{k}\) subproblem. 
Applying this argument recursively for $k = \numlevels^i-1, ..., 1$ shows that player $i$'s candidate solution is weakly locally minimizing at every level of its hierarchy. 
Repeating this argument for all players \(i\in[\numplayers]\) yields that \(\left(\var^*,\eta^{1*}_{1:\numlevels^1},\dots,\eta^{\numplayers*}_{1:\numlevels^\numplayers}\right)\) is a weak local solution of the GOOP problem.
\end{proof}

% \textbf{Underlying constrained problem at level $k$}
% \begin{align}
% \label{eq:reduced-system-prob-level-k}
%     \min_{z_k^i} &\quad\cost{i}_{k}(\var^i_k, \var^{\neg i}_k) \\
%     \quad \st &\mathcal{F}_{k+1}(\var^i_k, \var^{\neg i}_k, \eta_{k+1}^i) = 0, \quad g^i(\var^i_k, \var^{\neg i}_k) \geq 0.
% \end{align}

% \textbf{Critical cone/directions for the problem \cref{eq:reduced-system-prob-level-k}}
% \begin{itemize}
%     \item Linearized feasible cone $L_F(z) \coloneqq \{d: \nabla_{\var}\mathcal{F}_{k+1}^\top d = 0, \nabla_{\var}g_j^{i\top} d = 0, j\in\mathcal{A}(\var)\}$
%     \item Critical cone $\mathcal{C}_k(\var, \lambda^i_k, \gamma^i_k) \coloneqq \{d\in L_F(z):  \nabla_{\var}g_j^{i\top} d = 0, j\in\mathcal{A}(\var), \gamma^i_{k,j} > 0\}$
% \end{itemize}

\begin{remark}[Implications of a strict local minimizer]
\label{rem:strict-minimizer-propagates}
If the inequality in \cref{eq:local-psd-primal} holds strictly at level $k$, the candidate point corresponds to a \emph{strict} local minimizer at that level.
% This implies that the lower-level solution is locally isolated.
Thus, in a neighborhood of the candidate point, the feasible set for the upper levels $k-1, \dots, 1$ collapses to the unique solution of the lower level.
As a result, the verification of sufficient conditions at these upper levels becomes unnecessary.
\end{remark} 

\begin{remark}[Limitations of \cref{thm:sosc-goop}]
    Per \Cref{rem:strict-minimizer-propagates}, the second-order condition established in \cref{thm:sosc-goop} checks whether a candidate solution is a weak local minimizer at each level for each player. Unfortunately, the condition is not easily checked, in general. This is true even in the case of single-level optimization as discussed in \cite{fiacco1968}. Establishing easily verified conditions for this case is an important direction for future work. 
\end{remark}
In the next section, we develop a primal-dual interior-point method and analyze its convergence properties. 
This common solver framework enables a controlled comparison in \Cref{sec:experiment}, where we quantify the computational benefits of solving the reduced system relative to the complete formulation.

\section{Primal-Dual Interior Point Method for Reduced KKT System and Convergence Analysis}
% \begin{itemize}
%     \item Regularized Primal-Dual KKT System
%     \item Algorithmic Implementation
%     \begin{itemize}
%         \item algorithmic updates (Newton direction, relaxation, etc.)
%     \end{itemize}
%     \item Convergence Analysis
%     \begin{itemize}
%         \item Reference: Bertsekas, Nocedal 
%         \item Could be split into ``global" and ``local" convergence subsectinos, i.e., global convergence under convexity(?), local superlinear rate under standard assumptions
%         \item S\&B pg 536
%     \end{itemize}
% \end{itemize}

In this section, we present a numerical approach for solving the reduced KKT conditions. 
We base our approach on the primal-dual interior point (PDIP) method \cite{nocedal2006numerical}, which we view as a homotopy method \cite{liao2004homotopy}.
To this end, we (i) introduce nonnegative slack variables $\{ s_k^i\}_{k=1,i=1}^{\numlevels^i,\numplayers}$, which transform the inequality constraints into equality constraints for all players $i \in [\numplayers]$ and levels $k \in [\numlevels^i]$, 
and (ii) introduce a scalar homotopy parameter $\rho > 0$ to perturb the complementarity slackness condition for each player $i$ at each level $k$:
\begin{equation}
\label{pdip-slack-variables}
   \inequality^i(\var^{i}, \var^{\neg i}) - s_k^i = \bm{0}, \;\;
 s^i_k \odot \gamma^i_k = \rho \bm{1}, \;\; s^i_k \geq \bm{0}, \;\; \gamma^i_k \geq \bm{0},
\end{equation}
where $\bm{1}$ denotes a vector of all ones.
Player $i$'s perturbed reduced KKT system is
\begin{align}
\label{eqn:reduced-kkt-top-level-with-slack}
\mathcal{K}_\rho^i\left(\var^{i}, \var^{\neg i}, \eta_{1:\numlevels^i}^i, s^i_{1:\numlevels^i}\right) 
&= 
\begin{bNiceMatrix}[margin,cell-space-limits=1.5pt]
    \nabla_{\var^i} \Lagrangian_{1}^i(\var^{i}, \var^{\neg i}, \eta_{1:\numlevels^i}^i) \\[-1.5ex]
    \vdots \\
    \nabla_{\var^{i}} \Lagrangian_{\numlevels^i}^i(\var^{i}, \var^{\neg i}, \eta_{\numlevels^i}^i) \\
    \equality^i(\var^{i}, \var^{\neg i})\\
    \inequality^i(\var^{i}, \var^{\neg i}) - s_1^i \\[-1.5ex]
    \vdots \\
    \inequality^i(\var^{i}, \var^{\neg i}) - s_{\numlevels^i}^i \\
    s_{1:\numlevels^i}^i \odot \gamma_{1:\numlevels^i}^i - \rho \bm{1}
\end{bNiceMatrix} = \bm{0}.
% \mathcal{G}^i\left( \var^{i}, \var^{\neg i}, \gamma_{1:\numlevels^i}^i\right)
% &= 
% \begin{bNiceMatrix}[margin]
% \inequality^i(\var^{i}, \var^{\neg i}) \\
%     \gamma_{1:\numlevels^i}^i
% \end{bNiceMatrix} \geq 0.
\end{align}

Aggregating \cref{eqn:reduced-kkt-top-level-with-slack} over all players and defining 
% \david{why bold here but not up to this point?} 
% \david{also, i feel like you can just put (6.4) inline} 
% ${\var} = \left[(\var^1)^\top;\dots;(\var^\numplayers)^\top\right]^\top$, 
${\eta} \coloneqq [(\eta^1_{1:\numlevels^1})^\top,\dots,(\eta^\numplayers_{1:\numlevels^\numplayers})^\top]^\top$, 
${s} \coloneqq [(s^1_{1:\numlevels^1})^\top,\dots,(s^\numplayers_{1:\numlevels^\numplayers})^\top]^\top$, and  
${y} \coloneqq [\var^\top, \eta^\top, s^\top]^\top$, we obtain the $\rho$-perturbed reduced KKT system: 
\(\mathcal{K}_\rho\left({y}\right) = \Bigl[
    \mathcal{K}_\rho^i\left(\var^{i}, \var^{\neg i}, \eta_{1:\numlevels^1}^i, s^i_{1:\numlevels^1}\right)
    % \vdots \\
    % \mathcal{K}_\rho^\numplayers\left(\var^{\numplayers}, \var^{\neg \numplayers}, \eta_{1:\numlevels^\numplayers}^\numplayers, s^\numplayers_{1:\numlevels^\numplayers}\right)
\Bigr]_{i=1}^{\numplayers} = \bm{0}.\)
We present our method in \cref{alg:goop-pdip}. 
The algorithm gradually decreases homotopy parameter \(\rho\) to zero. 
For each \(\rho\), at iteration \(\ell\), we compute \({y}_\rho^{(\ell)}\) that drives the homotopy residual to zero, i.e.,  $\|\mathcal{K}_\rho({y}_\rho^{(\ell)})\|_2 \to 0$.
Since $\nabla \mathcal{K}_\rho({y}_\rho^{(\ell)})$ is not a square matrix in general, we compute the Newton update direction $\Delta {y}_\rho$ using the pseudoinverse $(\nabla \mathcal{K}_\rho({y}_\rho^{(\ell)}))^+$,
\begin{equation}
\Delta {y}_\rho
:= (\nabla \mathcal{K}_\rho({y}_\rho^{(\ell)}))^+(-\mathcal{K}_\rho({y}_\rho^{(\ell)})).
\end{equation}
We then choose a step size \(\alpha^{(\ell)} \in (0,1]\) via backtracking line search and update
\begin{equation}
{y}_\rho^{(\ell+1)}
= {y}_\rho^{(\ell)} + \alpha^{(\ell)} \Delta {y}_\rho.
\end{equation}
During line search, we ensure that the slack variables \(s^{(\ell)}_\rho\) and dual variables \(\gamma^{(\ell)}_\rho\) remain nonnegative after the update.
Using the KKT residual \(\|\mathcal{K}_\rho({y}_\rho^{(\ell)})\|_2\) as the merit function, we repeat this procedure until convergence, after which $\rho$ is reduced and the process continues. 
Next, we characterize regularity conditions under which our method is provably convergent in \cref{thm:quadratic-convergence}.
% \begin{algorithm}[t]
% % \caption{PDIP-LQ Games}
% \caption{Local GOOP Equilibrium via PDIP}
% \begin{algorithmic}[1]\label{alg:goop-pdip}
% \REQUIRE %$\{A_t,B_t,c_t\}_{t=0}^T$
% Initial homotopy parameter $\rho$, contraction rate $\sigma\in(0,1)$, 
% tolerance $\epsilon$, initial solution ${y}_\rho^{(0)}:=[{\var}_\rho^{(0)},{\eta}_\rho^{(0)}, {s}_\rho^{(0)}]$ with $\svec_\rho^{(0)}>0$ and $\gammavec_\rho^{(0)}>0$
% \FOR{$\ell = 1,2,\dots,\ell_{\mathrm{max}}$}
% \WHILE{$\|\mathcal{K}_\rho({y}_\rho^{(\ell)})\|_2 > \epsilon $}
% \STATE Initialize the step size for line search, $\alpha \gets 1$
% \STATE Compute Newton update direction $\Delta {y}_\rho$ from $ 0 = \nabla \mathcal{K}_\rho({y}_\rho^{(\ell)}) \Delta {y}_\rho + \mathcal{K}_\rho({y}_\rho^{(\ell)})$
% \STATE Compute step size $\alpha$ via fraction-to-the-boundary linesearch
% \IF{$\alpha < \epsilon $}
% \STATE Claim line search \textbf{failure}
% \ENDIF
% \STATE ${y}_\rho^{(\ell+1)}\gets {y}_\rho^{(\ell)}+\alpha \Delta {y}_\rho$
% \ENDWHILE
% \STATE  $\rho \gets \sigma \cdot \rho$
% \STATE
% \ENDFOR
% \STATE Record ${y}_\rho \gets {y}_\rho^{(\ell)}$. 
% \RETURN ${y}_\rho$
% \end{algorithmic}
% \end{algorithm}

% \david{ suggest matching alg formatting closer to what i have in my class notes - i like that aesthetic a bit better. i just got it by looking up ``latex algorithms overleaf'' and scrolling down through docs}
{\footnotesize
\begin{algorithm}[t]
\caption{Primal-Dual Interior Point Method for the Reduced KKT System}
\begin{algorithmic}[1]\label{alg:goop-pdip}
\REQUIRE Initial homotopy parameter \(\rho\), tolerance \(\epsilon\), maximum outer iterations \(\ell_{\max}\), line search parameter \(\beta \in (0,1)\), contraction rate \(\sigma \in (0,1)\),
initial solution \({y}_\rho^{(0)} := [{\var}_\rho^{(0)},{\eta}_\rho^{(0)}, {s}_\rho^{(0)}]\) with \({s}_\rho^{(0)} > \bm{0}\) and \({\gamma}_\rho^{(0)} > \bm{0}\).

\FOR{\(\ell = 0,1,2,\dots,\ell_{\max}-1\)}
    \WHILE{\(\|\mathcal{K}_\rho({y}_\rho^{(\ell)})\|_2 > \epsilon\) 
    % \jingqi{A simple algorithm logic problem: what if the while loop never ends? I think replacing the while with a for loop can resolve this for the inner update, i.e., fixing rho, solve the solution satisfying $K_{\rho}$ }
    }
        \STATE Compute the Newton update direction $\Delta {y}_\rho = - (\nabla \mathcal{K}_\rho({y}_\rho^{(\ell)}))^+ \mathcal{K}_\rho({y^{(\ell)}_\rho})$
        % $\nabla \mathcal{K}_\rho({y}_\rho^{(\ell)})\,\Delta {y}_\rho = -\mathcal{K}_\rho({y}_\rho^{(\ell)})$
        % \jingqi{Notice that there are multiple solutions to this equation, and we should specify $\Delta y = - \nabla \mathcal{K}_\rho ({y})^+ K_\rho(y)$. If we don't use consistent min-norm solution, then we don't have convergence guarantee in experiment}
        \STATE Initialize step size \(\alpha \gets 1\)
        % \STATE Compute \(\alpha\) via backtracking line search while ensuring ${s}_\rho^{(\ell)} > 0$, ${\gamma}_\rho^{(\ell)} > 0$
        \WHILE{ \(\|\mathcal{K}_\rho({y}_\rho^{(\ell)} + \alpha \Delta {y}_\rho)\|_2 > \|\mathcal{K}_\rho({y}_\rho^{(\ell)})\|_2\)
        or \(\hat{{y}} \coloneqq {y}_\rho^{(\ell)} + \alpha \Delta {y}_\rho \) has a nonpositive element in its subvector \([\hat{{s}}_\rho,\hat{{\gamma}}_\rho]\)}
        \STATE \(\alpha \gets \beta \cdot \alpha\)
        \ENDWHILE
        \STATE \textbf{If} \(\alpha < \epsilon\), \textbf{then declare ``line-search failure" and break.}  
        % \IF{\(\alpha < \epsilon\)}
        %     \STATE \textbf{break} \hfill \COMMENT{line search failure}
        % \ENDIF
        \STATE \({y}_\rho^{(\ell)} \gets {y}_\rho^{(\ell)} + \alpha \Delta {y}_\rho\).
    \ENDWHILE
    \STATE \({y}_\rho^{(\ell+1)} \gets {y}_\rho^{(\ell)}\) and $\rho \gets \sigma\cdot\rho$
    % \jingqi{Two things: 1. the current algorithm updates rho every iteration, and this will destabilize the update because the KKT residual is not consistenly defined with a fixed rho, and it is different from the convergence analysis, also diverging from most prior PDIP works, which adopt two-loop algorithm: in the outer loop, we periodically decaying rho and in the inner loop we do the Newton's step with rho fixed; 2. Warmstart the initial solution for the next homotopy parameter $\rho$ value $ {y}_{\sigma \rho}^{(0)} \gets {y}_\rho^{(\ell)} $, and update $\rho \gets \sigma\cdot\rho$}
    % \IF{\(\|\mathcal{K}_\rho({y}_\rho^{(\ell+1)})\|_2 \le \epsilon\)} 
    %     \STATE \(\rho \gets \rho\bigl(1 - e^{(-\kappa_{\mathrm{tight}}\, m)\bigr)}\) \hfill \COMMENT{tighten}
    % \ELSE
    %     \STATE \(\rho \gets \rho\bigl(1 + e^{(-\kappa_{\mathrm{loose}}\, m)\bigr)}\) \hfill \COMMENT{loosen}
    % \ENDIF
    % \STATE \(\rho \gets \min(\rho, 1)\) \label{alg:min-rho}
\ENDFOR
\RETURN \({y}_\rho^{(\ell)}\)
\end{algorithmic}
\end{algorithm}
}

\begin{theorem}[Local quadratic convergence]
\label{thm:quadratic-convergence}
    Let ${y}^{(0)}_\rho$ be an initial point with ${\gamma}^{(0)}_\rho \odot {s}^{(0)}_\rho \ge \epsilon \bm{1}$, for some $\epsilon \in (0,\rho)$. Define $\mathcal{S}_{y_\rho} := \{{y_\rho}: \| \mathcal{K}_\rho({y_\rho})\|_2 \le \|\mathcal{K}_\rho({y}^{(0)}_\rho)\|_2, \gamma_\rho\odot {s}_\rho\ge \epsilon \bm{1} \}$ and $L := \|\mathcal{K}_\rho({y_\rho}^{(0)})\|_2$. Suppose that $ \mathcal{K}_\rho(y_\rho)$ is in the column space of $\nabla \mathcal{K}_\rho({y_\rho})$ for all ${y_\rho}\in \mathcal{S}_{y_\rho} $ and that there exists an optimal solution $ {y^*_\rho} \in \mathcal{S}_{y_\rho}$ such that $\mathcal{K}_\rho({y}^*_\rho) = 0$. Furthermore, suppose that there exist constants $C,D>0$ such that the pseudoinverse of $\nabla \mathcal{K}_\rho({y}_\rho)$ is upper bounded by $D$, $\|\big(\nabla \mathcal{K}_\rho({y}_\rho)\big)^{+}\|_2\le D,  \forall {y_\rho}\in \mathcal{S}_{y_\rho}$, and the Jacobian $\nabla \mathcal{K}_\rho({y}_\rho)$ is $C$-Lipschitz continuous:
    \begin{equation}
        \|\nabla \mathcal{K}_\rho({y_\rho}) - \nabla \mathcal{K}_\rho(\tilde{y}_\rho)\|_2 \le C\|{y_\rho} - \tilde{{y}}_\rho\|_2, \forall {y}_\rho, \tilde{y}_\rho \in \mathcal{S}_{y_\rho}.
    \end{equation}
    Define $\hat{\alpha}:=\min\big\{1,\frac{1}{CD^2 L}, \frac{\rho-\epsilon}{D^2 L^2}\big\}$, and $\Delta {y_\rho}= -(\nabla \mathcal{K}_\rho({y_\rho}))^+ \mathcal{K}_\rho({y_\rho})$. Then, for all ${y_\rho}\in\mathcal{S}_{y_\rho}$, there exists a stepsize $\alpha \in(0,1]$ such that ${y_\rho} + \alpha\Delta {y_\rho}\in \mathcal{S}_{y_\rho}$ and 
    \begin{equation}\label{eq: linear convergence}
        \|\mathcal{K}_\rho({y_\rho}+ \alpha \Delta {y_\rho})\|_2\le \left(1-\frac{\hat{\alpha}}{2}\right)\|\mathcal{K}_\rho({y_\rho})\|_2.
    \end{equation}
    Moreover, when $\|\mathcal{K}_\rho({y_\rho})\|_2\le \min\big\{\frac{2}{CD^2}, \frac{\sqrt{\rho-\epsilon}}{D}\big\}$, we have ${y_\rho}+\Delta {y_\rho}\in \mathcal{S}_{y_\rho}$ and quadratic convergence holds: 
    % \david{maybe cleaner to divide out by the common fraction on the left and right}
    \begin{equation}
        \frac{D^2C}{2}\|\mathcal{K}_\rho({y_\rho}+\Delta {y_\rho})\|_2\le \left(\frac{D^2C}{2}\|\mathcal{K}_\rho({y_\rho})\|_2\right)^2.
    \end{equation}
\end{theorem}
\begin{proof}
    The proof is provided in the Appendix. 
\end{proof}

Given a fixed $\rho>0$, $\mathcal{K}_\rho$ and $\mathcal{K}_0$ differ only in the complementarity conditions, and consequently the proximity between ${y}_\rho^*$ and ${y}_0^*$ is directly controlled by $\rho$. We characterize the solution error between the converged solution ${y}_\rho^*$ and a ground truth solution ${y}_0^*$ to the reduced KKT system in the following result. 
\begin{theorem}[Central path]\label{thm:central path}
    Let ${y}_0^*$ satisfy $\mathcal{K}_0({y}_0^*)=0$ under Assumption \ref{assum:cost-feasible-set}, with $\| \nabla \mathcal{K}_0({y})^+\|_2\le D_0 < \infty$ and $\nabla \mathcal{K}_0({y})$ having constant row rank in a neighborhood of ${y}_0^*$. Let $N_c = \sum_{i=1}^N  m_{\mathcal{I}}^i$ be the total number of complementarity pairs. 
    % \david{not surue what you mean ``pairs?''}. 
    Then, the solution ${y}_\rho^*$ to $\mathcal{K}_\rho({y}_\rho^*)=0$ satisfies
    \begin{equation}
        \|{y}_\rho^* - {y}_0^*\|_2\le D_0 \sqrt{N_c} \rho.
    \end{equation}
\end{theorem}
\begin{proof}
    The proof is provided in the Appendix. 
\end{proof}

Specifically, Theorem~\ref{thm:central path} shows that $\| {y}_\rho^* - {y}_0^* \|_2 = \mathcal{O}(\rho)$, meaning that the error between the converged PDIP iterate and the solution of the reduced KKT system at $\rho=0$ decomposes into two independent contributions: 
\begin{equation}
    \| {y}_\rho^{(\ell)} - {y}_0^* \|_2 \le \| {y}_\rho^{(\ell)} - {y}_\rho^* \|_2 + D_0 \sqrt{N_c}\rho
\end{equation}
where the first term is the algorithmic error, controlled by the PDIP convergence tolerance ${\epsilon}$ and decreasing quadratically by Theorem~\ref{thm:quadratic-convergence}, and the second term is the perturbation error, arising from the $\rho$-relaxation of the complementarity condition ${s} \odot {\gamma} = \rho \bm{1}$ and decreasing linearly with $\rho$. Therefore, driving both $\epsilon$ and $\rho$ to zero is sufficient to recover a ground truth solution to the reduced KKT conditions.

\section{Numerical Study}
\label{sec:experiment}
% \begin{itemize}
%     \item Two-player intersection scenario with four preference levels
%     \item Solver Performance Comparison: computation times, solution quality 
% \end{itemize}

% {\color{blue}
% Other potential theoretical contributions:
% \begin{enumerate}
%     \item How fast should we decrease the annealing parameter $\mu$ in PDIP such that we have consistent convergence across each parameter annealing step without diverging? 
%     \item Theoretically, would there be a benefit to using a degrading/increasing step size across different layers? 
%     \item What conditions of the problem objectives / constraints can guarantee the exponential convergence? 
% \end{enumerate}
%}

In this section, we report a numerical study\footnote{Code is available at \url{https://github.com/CLeARoboticsLab/Reduced-GOOP}.} to quantify the advantages of the reduced KKT system formulation compared with the complete formulation.
% computational behavior of our proposed reduced KKT system for GOOP and to quantify its practical advantages over the complete KKT construction. 
% We consider two representative regimes: (i) quadratic GOOPs, as in \Cref{subsec:equivalence-quadratic-goop}, and (ii) nonquadratic GOOPs, obtained by adding a nonquadratic objective term at the innermost preference level.

\subsection{Complexity Study}
We compare reduced and complete KKT systems on quadratic and nonquadratic GOOPs with $N=4$ players. Each player $i$ has $n^i = 10$ primal variables, $\equaldim^i = 3$ linear equality constraints, and $\inequaldim = 2$ linear inequality constraints. We vary the number of preference levels $K \in \{2,\dots,6\}$ and, for each $K$, generate 100 random instances with initial values (for primal variables) perturbed as $\mathcal{N}(0,1)$ around a feasible primal solution $z_\rho^{(0)}$. In the quadratic setting, each $Q_k^i$ in \cref{eq:quadratic-cost} is sampled as rank-2 positive semidefinite from Gaussian matrices, and $q_k^i$ is chosen from the column space of $Q_k^i$ to ensure boundedness; constraints satisfy \cref{eq:inequality-licq}. 
To generate nonquadratic problem instances, we replace the outermost objective with $(\mathbf{1}^\top z)^4$. We adapt Algorithm~\ref{alg:goop-pdip} to the complete KKT system by approximating complementarity as in \eqref{pdip-slack-variables}, and run it on both the reduced and complete formulations using a geometric schedule $\rho \in\{1,10^{-1},\dots,10^{-10}\}$. 

\begin{table}[ht]
\caption{Comparison of reduced and complete KKT systems across preference levels. \textnormal{R} denotes the reduced KKT system, and \textnormal{C} denotes the complete KKT system. }
\label{tab:kkt-compare-levels}
\centering
\setlength{\tabcolsep}{4pt}
\renewcommand{\arraystretch}{1.05}
\resizebox{\linewidth}{!}{%
\begin{tabular}{c c c c c c c c c}
\toprule
& \multicolumn{2}{c}{\shortstack{System\\size}}
& \multicolumn{2}{c}{\shortstack{Variable\\size}}
& \multicolumn{2}{c}{\shortstack{Solve time (s)\\for Quad.\ GOOP}}
& \multicolumn{2}{c}{\shortstack{Solve time (s)\\for Nonquad.\ GOOP}} \\
\cmidrule(lr){2-3}\cmidrule(lr){4-5}\cmidrule(lr){6-7}\cmidrule(lr){8-9}
Level ($K$) & R & C & R & C & R & C & R & C \\
\midrule
2 & 132  & 168  & 128  & 136  & $0.05 \pm 0.06$ & $0.07 \pm 0.06$ & $0.08 \pm 0.12$ & $0.19 \pm 0.21$ \\
3 & 188 & 368  & 244 & 304  & $0.17 \pm 0.16$ & $0.34 \pm 0.12$ & $0.53 \pm 0.44$ & $1.11 \pm 0.35$ \\
4 & 244 & 800  & 408 & 672  & $0.32 \pm 0.32$ & $2.32 \pm 2.20$ & $0.54 \pm 0.56$ & $5.09 \pm 2.82$ \\
5 & 300 & 1728 & 620 & 1472 & $0.60 \pm 0.75$ & Failed          & $0.96 \pm 0.88$ & Failed \\
6 & 356 & 3712 & 880 & 3200 & $1.19 \pm 0.80$ & Failed          & $1.79 \pm 3.01$ & Failed \\
% 2 & 132  & 168  & 128  & 136  & $0.02 \pm 0.02$ & $0.04 \pm 0.02$ & $0.05 \pm 0.01$ & $0.09 \pm 0.12$ \\
% 3 & 188 & 368  & 244 & 304  & $0.12 \pm 0.06$ & $0.39 \pm 0.16$ & $0.23 \pm 0.20$ & $0.46 \pm 0.13$ \\
% 4 & 244 & 800  & 408 & 672  & $0.19 \pm 0.05$ & $1.72 \pm 0.25$ & $0.39 \pm 0.17$ & $3.30 \pm 3.76$ \\
% 5 & 300 & 1728 & 620 & 1472 & $0.41 \pm 0.74$ & Failed          & $0.72 \pm 0.61$ & Failed \\
% 6 & 356 & 3712 & 880 & 3200 & $1.10 \pm 1.83$ & Failed          & $1.02 \pm 37.85$ & Failed \\
% 2 & 132  & 168  & 128  & 136  & $0.04 \pm 0.02$ & $0.05 \pm 0.02$ & $0.05 \pm 0.01$ & $0.13 \pm 0.12$ \\
% 3 & 188 & 368  & 244 & 304  & $0.22 \pm 0.06$ & $0.43 \pm 0.16$ & $0.29 \pm 0.20$ & $0.50 \pm 0.13$ \\
% 4 & 244 & 800  & 408 & 672  & $0.21 \pm 0.05$ & $1.73 \pm 0.25$ & $0.43 \pm 0.17$ & $4.48 \pm 3.76$ \\
% 5 & 300 & 1728 & 620 & 1472 & $0.69 \pm 0.74$ & Failed          & $0.89 \pm 0.61$ & Failed \\
% 6 & 356 & 3712 & 880 & 3200 & $1.58 \pm 1.83$ & Failed          & $12.98 \pm 37.85$ & Failed \\
\bottomrule
\end{tabular}%
}
\end{table}
\cref{tab:kkt-compare-levels} reports: (i) total wall-clock solve time of Algorithm \ref{alg:goop-pdip} summed over all values of $\rho$ per instance, reported as mean ± standard deviation over 100 random instances after trimming the top and bottom $2.5\%$ of runtimes; (ii) the number of variables in each KKT system, $N_{\mathrm{\var^i,\eta^i}}$ (reduced) and $N_{\mathrm{\bar\var^i,\bar\eta^i}}$ (complete); and (iii) the number of equality and inequality constraints appearing in each KKT system, $({\mathcal{F}}^i,{\mathcal{G}}^i)$ (reduced) and $(\bar{\mathcal{F}}^i,\bar{\mathcal{G}}^i)$ (complete).
The complete system exhibits exponential growth in system size with respect to $K$, reflected by the rapid increase in the number of variables and constraints. 
This exponential growth translates directly into greater computational burden: across all tested preference levels, the complete formulation requires substantially longer solve times, and the performance gap widens as $\numlevels$ increases.
The same qualitative behavior is observed in both the quadratic and nonquadratic settings. 
We do not report nonquadratic GOOP results for the complete KKT system formulation when considering larger preference hierarchies ($K>4$), since the symbolic compilation became intractable due to the large size of complete KKT systems. 
% at this scale.
% complete KKT systems become computationally prohibitive at these sizes.

\begin{figure}[t]
    \centering
    \includegraphics[width=\linewidth]{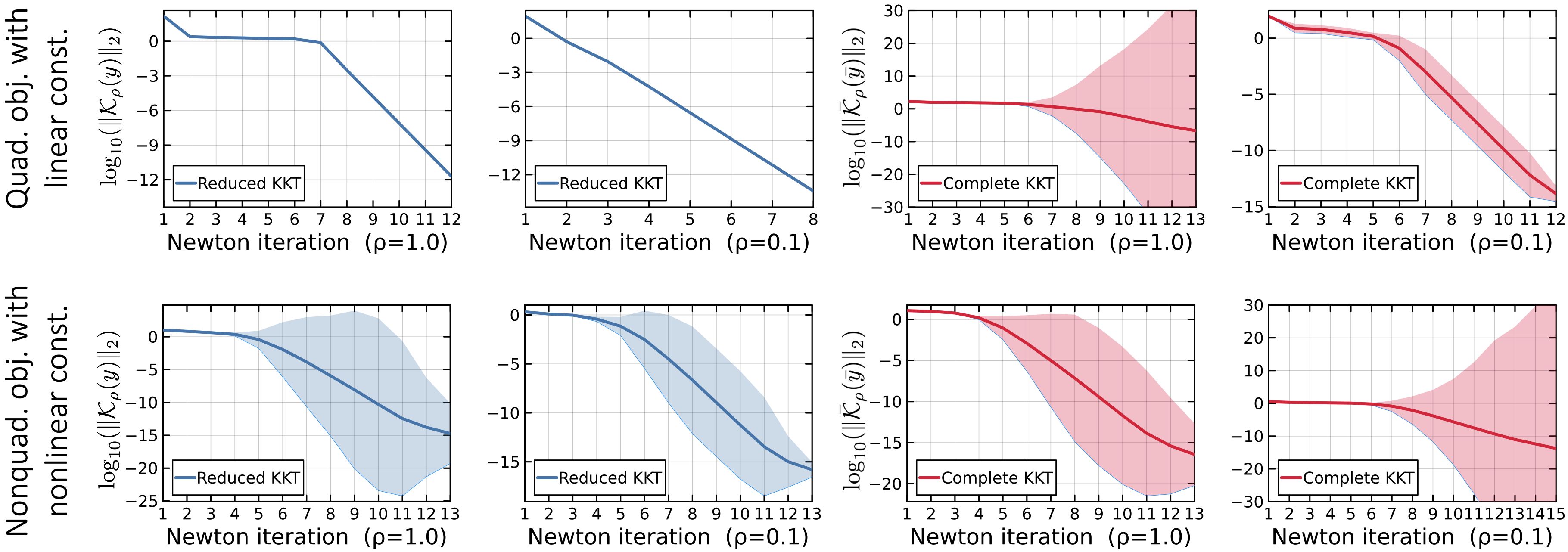}\vspace{-0.5em}
    \caption{Convergence of Algorithm~\ref{alg:goop-pdip} for reduced and complete KKT systems under varying $\rho$. The curve and the shaded area denote the mean and the variance of $\log_{10}(\|\mathcal{K}_\rho(y)\|_2)$, respectively. 
}
    \label{fig:trace}
\end{figure}
\begin{figure}[t]
\centering
\begin{minipage}{0.49\textwidth}
    \centering
    \includegraphics[width=\linewidth]{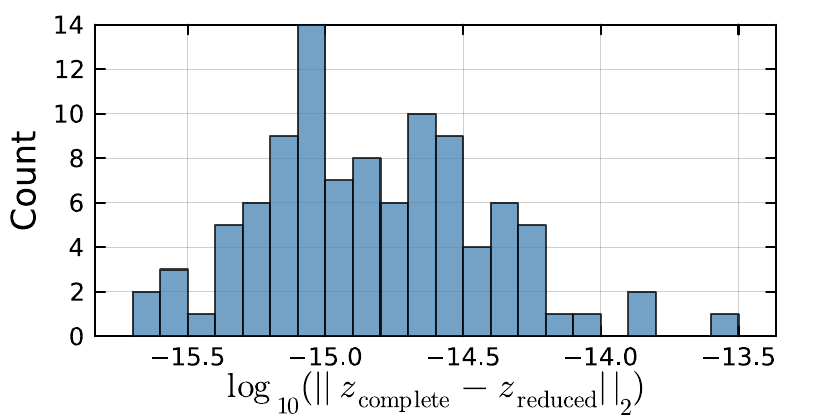}
\end{minipage}
\hfill
\begin{minipage}{0.5\textwidth}
    \caption{Monte Carlo comparison over 100 randomly generated GOOP instances with nonquadratic objectives and nonlinear constraints, showing close agreement between primal solutions $z_{\mathrm{complete}}$ and $z_{\mathrm{reduced}}$, obtained from the complete and reduced KKT systems, respectively.
    % Monte Carlo comparison of primal solutions from complete and reduced KKT conditions over 100 randomly generated GOOP problems with nonquadratic objectives and nonlinear constraints, where the primal solutions $z_{\textrm{complete}}$ under the complete KKT systems are close to the primal solutions $z_{\textrm{reduced}}$ under the reduced KKT systems.
 %Empirically identical, though theoretical guarantees are hindered by the complexity of the nonlinear KKT system's solution space.
}\label{fig:monte}
\end{minipage}
\end{figure}
% Finally, we observe primal exactness across all reported instances. 
% Specifically, for each primal solution of the reduced KKT system, we recover dual variables so that the complete KKT conditions are satisfied at the \emph{same} primal point.
% Thus, in both the quadratic and nonquadratic regimes, the reduced KKT system provides a computationally efficient route to finding identical primal solutions.
% This shows that we can effectively reduce system size growth from exponential to polynomial in the formulation. 

% \begin{remark}
% For quadratic \ac{goop}, this empirical primal equivalence is supported by our analysis in \cref{thm:ng-og-primal-equivalence,thm:ng-og-primal-equivalence-inequality}. 
% In the nonquadratic regime, we likewise observe primal solution equivalence across all tested instances, even though the quadratic proof strategy does not directly extend (as discussed earlier, dual variable misalignment can enter the upper-level stationarity conditions through higher-order terms). 
% Indeed, our only theoretical result in this case (\cref{thm:reduced-is-relaxation-of-complete}) ensures that primal solutions to the complete KKT system are also primal solutions to the reduced KKT system (but not necessarily vice versa).
% % \david{$\leftarrow$ add something like this, unless I am horribly mistaken about the result}
% Establishing conditions under which the reduced and complete KKT systems are guaranteed to share the same primal solutions in general settings therefore remains an important direction for future work.
% \end{remark}

\subsection{Convergence Study}
We empirically demonstrate linear and local quadratic convergence of Algorithm~\ref{alg:goop-pdip} on two problem classes. Class~(i):~$N=3$ players with $n^i=10$ variables and $K=4$ levels, using rank-2 quadratic objectives $J_k^i(z)=z^\top Q_{k}^i z + q_{k}^{i\top} z$ and linear constraints whose normals lie in the column space of $Q_{k}^i$. Class~(ii):~$N=2$ players with $n^i=4$ variables and $K=3$ levels, with a top-level $L_2$ objective $J_1^i(z)=\|z\|_2^2$ ensuring a unique equilibrium, inner objectives $J_k^i(z)=e^{v_{i,k}^\top z}+e^{-v_{i,k}^\top z}$ for random unit vectors $v_{i,k}$ (obtained by normalizing Gaussian samples), and nonlinear coupling constraints. In both settings, the convergence curves in Figure~\ref{fig:trace}, showing the mean and variance of $\log_{10}(\|\mathcal{K}_\rho(y)\|_2)$ over 10 runs with initializations perturbed as $\mathcal{N}(0,0.5)$ around a feasible solution $z_0$, empirically support Theorem~\ref{thm:quadratic-convergence}. Figure~\ref{fig:monte} further presents a Monte Carlo study over 100 instances of class~(ii), indicating that the reduced KKT system recovers the same primal solution as the complete system, although theoretical guarantees remain challenging due to the complexity of the nonlinear KKT solution space.

\subsection{Application to a practical scenario}
Finally, \cref{fig:2-player-intersection} demonstrates the applicability of the reduced GOOP formulation and Algorithm~\ref{alg:goop-pdip} on a practical multi-vehicle intersection planning problem, originally formulated in \cite{RAL-GOOP}, illustrating the potential of games of ordered preference in practical settings.

\begin{figure}[t]
    \centering
    \begin{subfigure}[t]{0.28\linewidth}
        \centering
        \includegraphics[width=\linewidth]{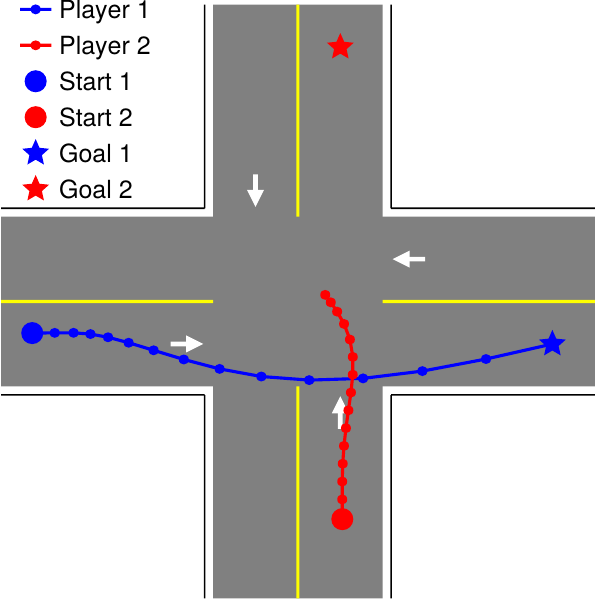}
        \caption{Computed trajectories}
        \label{fig:fig1}
    \end{subfigure}
    \hfill
    \begin{subfigure}[t]{0.34\linewidth}
        \centering
        \includegraphics[width=\linewidth]{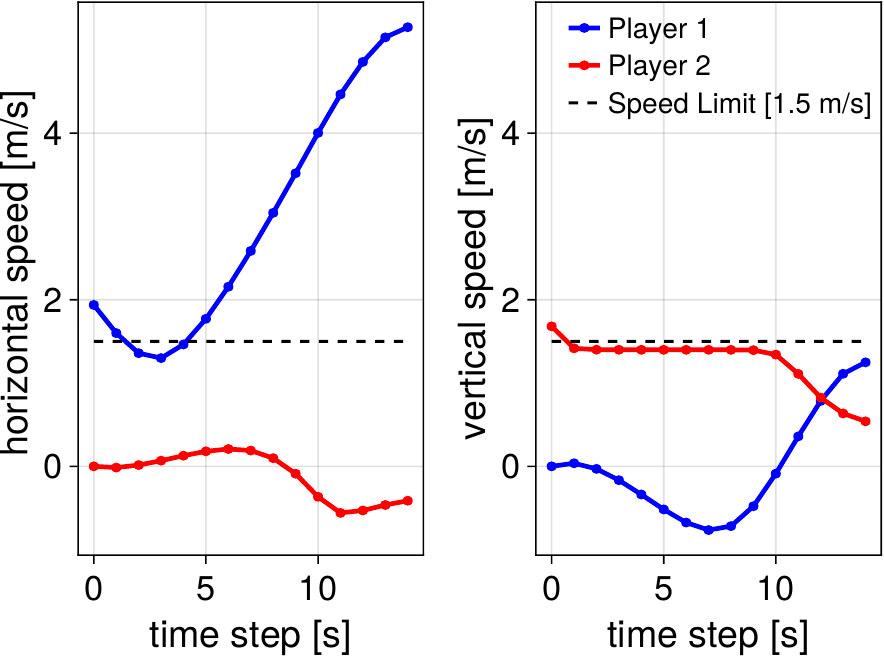}
        \caption{Velocity profile}
        \label{fig:fig2}
    \end{subfigure}
    \hfill
    \begin{subfigure}[t]{0.34\linewidth}
        \centering
        \includegraphics[width=\linewidth]{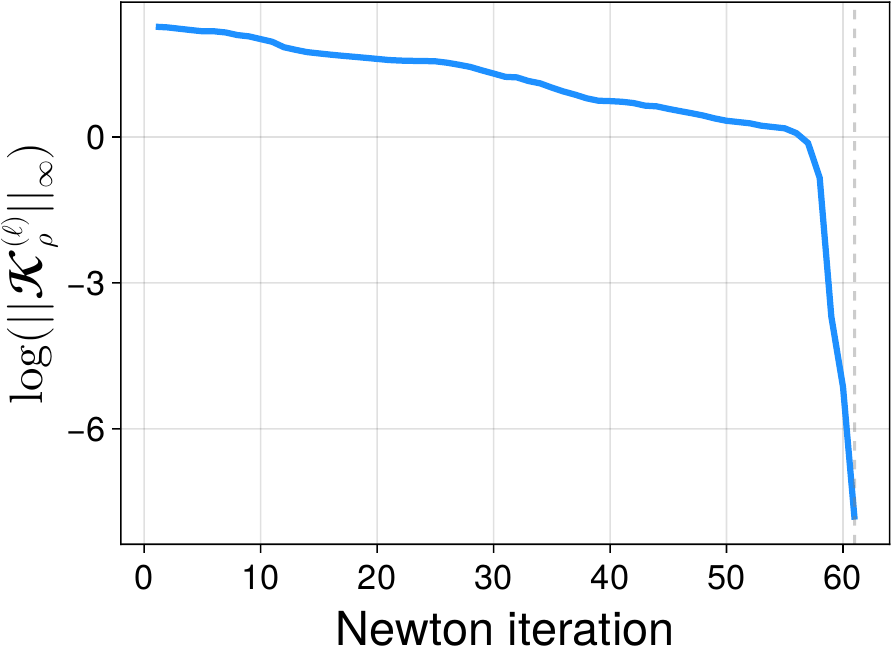}
        \caption{Convergence plot }
        \label{fig:fig3}
    \end{subfigure}
    \caption{Two-player intersection scenario solved using the reduced KKT system. We use the four-preference setup as in \cite{RAL-GOOP}. Player 1 (blue) prioritizes reaching its goal and accelerates beyond the speed limit. Player 2 (red) prioritizes obeying the speed limit and proceeds toward its goal while remaining within the limit. The convergence plot shows local quadratic convergence of the KKT residual to $10^{-8}$.
    % The homotopy parameter $\rho$ is initialized at $1.0$ and reduced by a factor of $\sigma = 0.5$ at each outer iteration according to \cref{alg:goop-pdip}; the reported results correspond to $\rho = 2^{-10}$, which corresponds to the minimum of the $\rho$ schedule.
    The homotopy parameter $\rho$ is reduced following a geometric schedule $\{1,2^{-1},\dots,2^{-10}\}$; the reported results correspond to $\rho = 2^{-10}$.}
    \label{fig:2-player-intersection}
\end{figure}
\section{Conclusion}
\label{sec:conclusion}

In this paper, we studied games of ordered preference and introduced a compact reduced KKT system that avoids the exponential complexity of standard single-level reformulations while preserving the essential primal stationarity structure across preference levels.
For quadratic GOOPs with linear constraints, we established primal solution equivalence between the reduced and complete KKT systems. For general smooth nonlinear GOOPs, we showed that the reduced system is a relaxation of the complete system and supplemented this result with second-order sufficient conditions for local optimality.
We developed a primal-dual interior-point method and proved its local quadratic convergence. Numerical study shows that the reduced KKT system formulation yields faster solve times than the complete KKT formulation, with benefits that grow in the number of preference levels.
Future work will address primal solution equivalence beyond the quadratic setting, and broaden the framework to richer classes of hierarchical games and multi-agent decision problems.
\appendix
\section{Supplementary results} 
\label{sec:appendix}

\begin{proof}[Proof of \cref{thm:exponential-growth-complete-kkt}]
For player $i$ at level $k$, we define
$
\bar d_k^i \coloneqq \dim(\bar\eta_k^i),
\bar p_k^i \coloneqq \dim(\bar\eta_{k+1:K^i}^i),
$
$\bar v_k^i \coloneqq \dim(\bar\var^i,\bar\eta_{k:K^i}^i),
\bar f_k^i \coloneqq \dim(\bar{\mathcal F}_k^i),
\bar g_k^i \coloneqq \dim(\bar{\mathcal G}_k^i)
$,
respectively denoting the numbers of dual variables, induced primals, total variables, equations, and inequality constraints.
For the innermost level, $\bar d_{\numlevels^i}^i = \equaldim^i + \inequaldim^i, \bar p_{\numlevels^i}^i = 0, \bar v_{\numlevels^i}^i = n^i + \equaldim^i + \inequaldim^i, \bar f_{\numlevels^i}^i = n^i + \equaldim^i + \inequaldim^i$, and $\bar g_{\numlevels^i}^i = 2\inequaldim^i$. 
At any $k\in[\numlevels^i-1]$, $\bar d_k^i = \bar f_{k+1}^i + \bar g_{k+1}^i$, $\;\bar p_{k}^i = \sum_{j=k+1}^{\numlevels^i}\bar d_j^i$, $\bar v_{k}^i = n^i + \bar p_k^i + \bar d_k^i$, $\bar f_{k}^i = n^i + \bar p_k^i + \bar d_k^i$, and $\bar g_{k}^i = 2^{\numlevels^i-k+1}\inequaldim^i$. 
Observe that $\bar d^i_k = \bar f_k^i - n^i - \bar p^i_k$ and $\bar p_{k-1}^i= \bar d_k^i + \bar p_k^i$ implies $\bar p_{k-1}^i = \bar f_k^i - n^i$.
Substituting this into $\bar f^i_{k-1} = n^i + \bar p^i_{k-1} + \bar f_{k}^i + \bar g_{k}^i$ yields $\bar f^i_k = 2\bar f_{k}^i + \bar g_{k}^i$. Then, applying the previous identity repeatedly and substituting $\bar f_{\numlevels^i}^i$ and $\bar g_{k+j}^i$,
\begin{align}
    \bar f_k^i &= 2\bar f_{k+1}^i + \bar g_{k+1}^i = 2\left(2\bar f_{k+2}^i + \bar g_{k+2}^i\right) + \bar g_{k+1}^i = \dots = 2^{\numlevels^i-k}\bar f_{\numlevels^i}^i + \sum_{j=1}^{\numlevels^i-k} 2^{j-1} \bar g_{k+j}^i \notag \\
          &= 2^{\numlevels^i-k}\left(n^i + \equaldim^i + (\numlevels^i-k+1)\inequaldim^i\right). \notag
\end{align}
\end{proof}

\begin{proof}[Proof of \cref{thm:reduced-is-relaxation-of-complete}]
We show that, for each player \(i\), every solution of the complete KKT system induces a solution of the reduced KKT system at the same primal point. The argument proceeds by backward induction on the preference level.

Fix a player \(i\), and let
$(\bar \var^i,\bar \eta^i_{1:\numlevels^i})$
satisfy the complete KKT system. Throughout the proof, we identify the primal variables in the two systems by setting $\var^i \coloneqq \bar \var^i.$
We then construct reduced-system dual variables \(\eta^i_{1:\numlevels^i}\) from the complete-system dual variables \(\bar\eta^i_{1:\numlevels^i}\) so that the reduced KKT conditions hold at the same primal point.

% At each level, the reduced system is obtained from the complete system by discarding those conditions involving stationarity with respect to induced lower-level dual variables and the associated complementarity relations. The corresponding (auxiliary) multiplier components in the reduced system are immaterial for feasibility and may therefore be chosen arbitrarily; we set them to zero
% \david{aren't they just not in the reduced system at all? i am confused why we are setting things to zero that are not even variables in the reduced system}. 

\paragraph{\textbf{Base case} \(k=\numlevels^i-1\)}
Consider the complete KKT system at level \(\numlevels^i-1\).
Let $\bar\gamma^i_{\numlevels^i-1,1}$ be the dual variable for the inequality constraint $g^i(\bar \var^i, \bar \var^{\neg i}) \geq \bm{0}$ and $\bar\gamma^i_{\numlevels^i-1,2}$ the dual variable for $\bar\gamma^i_{\numlevels^i} \geq \bm{0}$ constraint. 
Relative to the complete system $\bar{\mathcal{F}}^i_{\numlevels^i-1}$, the reduced system $\mathcal{F}^i_{\numlevels^i-1}$ omits (i)
$\nabla_{\bar\eta^i_{\numlevels^i}}\bar\Lagrangian^i_{\numlevels^i}
$
from the stationarity constraints and (ii) the complementarity conditions
$
\bar\gamma^i_{\numlevels^i}\odot \bar\gamma^i_{\numlevels^i-1,2}.
$ 
% \david{what is this last variable? i don't remember the comma 2 subscript}.
Accordingly, we set the reduced multipliers as
$
\psi^i_{\numlevels^i-1} = \bar\psi^i_{\numlevels^i-1}, \;
\phi^i_{\numlevels^i-1} = \bar\phi^i_{\numlevels^i-1}, \;
\lambda^i_{\numlevels^i-1} =  \bar\lambda^i_{\numlevels^i-1}, \;
\gamma^i_{\numlevels^i-1} = \bar\gamma^i_{\numlevels^i-1,1},
$
and discard the remaining dual variable \(\bar\gamma^i_{\numlevels^i-1,2}\) (associated with the omitted complementarity conditions).
Next, we obtain the Lagrangian of the reduced KKT system at level ${\numlevels^i-1}$ from the complete-system Lagrangian:
\begin{align}
\label{eqn:appendix-illustration-K-1}
    \bar \Lagrangian_{\numlevels^i-1}^i(\bar\var, \bar \eta^i_{\numlevels^i-1:\numlevels^i}) &= \cost{i}_{\numlevels^i-1}(\bar\var) - \bar\lambda_{\numlevels^i-1}^{i\top}h^i(\bar\var) - \bar\gamma_{\numlevels^i-1,1}^{i\top}g^i(\bar\var) \\ 
    &{\;\quad}- \bar\psi_{\numlevels^i-1}^{i\top}   \nabla_{\bar\var^i} \bar\Lagrangian^i_{\numlevels^i}
    - \bar\phi_{\numlevels^i-1}^{i\top} 
        \left(g^i(\bar\var) \odot \bar \gamma^i_{\numlevels^i}\right)
     - \cancel{\bar\gamma_{\numlevels^i-1,2}^{i\top}}  \bar \gamma^i_{\numlevels^i} \notag  \\ 
     &= \Lagrangian_{\numlevels^i-1}^i(\bar\var, \eta^i_{\numlevels^i-1:\numlevels^i}). \notag
\end{align}
Hence the reduced system \(\mathcal F^i_{\numlevels^i-1}\) is satisfied by the complete-system solution. Likewise, the reduced inequality system \(\mathcal G^i_{\numlevels^i-1}\) is obtained from \(\bar{\mathcal G}^i_{\numlevels^i-1}\) by removing the nonnegativity conditions associated with the discarded multipliers, i.e, \(\bar\gamma^i_{\numlevels^i-1,2}\geq\bm{0}\). 

\paragraph{\textbf{Inductive step}}
Let \(k\in\{2,\dots,\numlevels^i-1\}\), and suppose that reduced-system dual variables have already been constructed at levels \(k,\dots,\numlevels^i\) so that the reduced KKT conditions hold on those levels. We now construct the reduced dual variables at level \(k-1\).
Consider the complete KKT system at level \(k-1\). 
Let $\bar\gamma^i_{k-1,1}$ be the dual variable for the inequality constraint $g^i(\bar \var^i, \bar \var^{\neg i}) \geq \bm{0}$ and $\bar\gamma^i_{k-1,2}$ the dual variable for $\bar\gamma^i_{k:\numlevels^i} \geq \bm{0}$ constraint. In passing from the complete system $\bar{\mathcal{F}}^i_{k-1}$ to the reduced system ${\mathcal{F}}^i_{k-1}$, we omit: (i) the stationarity conditions with respect to lower-level dual variables, $\nabla_{\bar\eta^i_{k:\numlevels^i}} \bar\Lagrangian^i_{k-1}$, 
(ii) the complementarity conditions involving lower-level inequality multipliers, $\bar\gamma^i_{k:\numlevels^i}\odot \bar\gamma^i_{k-1,2}$.
% \david{same comment about maybe not introducing this variable before. also the cdot is hard to see}
% , together with the corresponding nonnegativity conditions on the associated multiplier components.
Next, we set the reduced dual variables at level \(k-1\) by retaining only a subset of components of \(\bar\eta^i_{k-1}\): (i) the components of \(\psi^i_{k-1}\) are taken from \(\bar\psi^i_{k-1}\) for the terms associated with $
    \nabla_{\var^i}\bar\Lagrangian^i_{k:\numlevels^i},
    % \nabla_{\var^i}\bar\Lagrangian^i_{k+1},\;
    % \dots,\;
    % \nabla_{\var^i}\bar\Lagrangian^i_{\numlevels^i},
    $
    (ii) the components of \(\phi^i_{k-1}\) are taken from \(\bar\phi^i_{k-1}\) for the terms associated with
    $
    g^i(\var^i,\var^{\neg i})\odot \bar\gamma^i_{k,1}, \dots, g^i(\var^i,\var^{\neg i})\odot \bar\gamma^i_{\numlevels^i}
    % g^i(\var^i,\var^{\neg i})\odot \bar\gamma^i_{k+1,1},\;
    % \dots,\;
    $
    and (iii) the current-level feasibility multipliers are preserved:
    $
    \lambda^i_{k-1} = \bar\lambda^i_{k-1}, \;
    \gamma^i_{k-1} =  \bar\gamma^i_{k-1,1}.
    $
The remaining complete-system dual variables at level \(k-1\), including $\bar\gamma^i_{k-1,2}$, are discarded. 
% \david{same comment as above. i think you are maybe trying to prove that a solution to the complete system can be constructed from the reduced system solution, which is the reverse of the theorem summary statement at the top of the proof}
We illustrate this construction below:
% \david{find a way to get this on the previous page} \david{I still don't follow the comma in the subscript terms } \david{why do this at level $k-1$ instead of $k$?}
\begin{align}
\label{eqn:appendix-illustration}
    &\bar \Lagrangian_{k-1}^i(\bar\var^i, \bar\var^{\neg i}, \bar \eta^i_{k-1:\numlevels^i}) = \cost{i}_{k-1}(\bar\var^i, \bar\var^{\neg i}) - \bar\lambda_{k-1}^{i\top}h^i(\bar\var^i, \bar\var^{\neg i}) - \bar\gamma_{k-1,1}^{i\top}g^i(\bar\var^i, \bar\var^{\neg i}) \\ 
    &- \bar\psi_{k-1}^{i\top} \begin{bmatrix}
        \nabla_{\bar\var^i} \bar\Lagrangian^i_k \\[0.5ex]
        \cancel{\nabla_{\bar\eta^i_{k+1:\numlevels^i}}\bar\Lagrangian^i_k} \\[0.5ex]
        \nabla_{\bar\var^i} \bar\Lagrangian^i_{k+1} \\[0.5ex]
        \cancel{\nabla_{\bar\eta^i_{k+2:\numlevels^i}}\bar\Lagrangian^i_{k+1}} \\
        \vdots \\
        \nabla_{\bar\var^i} \bar\Lagrangian^i_{\numlevels^i}
    \end{bmatrix}
    - \bar\phi_{k-1}^{i\top} \begin{bmatrix}
        g^i(\bar\var^i, \bar\var^{\neg i}) \odot \bar\gamma^i_{k,1} \\[0.5ex]
        \cancel{\bar\gamma^i_{k+1:\numlevels^i} \odot \bar\gamma^i_{k,2}} \\[0.5ex]
        g^i(\bar\var^i, \bar\var^{\neg i}) \odot \bar\gamma^i_{k+1,1} \\[0.5ex]
        \cancel{\bar\gamma^i_{k+2:\numlevels^i} \odot \bar\gamma^i_{k+1,2}} \\
        \vdots \\
        g^i(\bar\var^i, \bar\var^{\neg i}) \odot \bar \gamma^i_{\numlevels^i}
    \end{bmatrix} - \cancel{\bar\gamma_{k-1,2}^{i\top}} \begin{bmatrix}
        \bar \gamma^i_{k} \\
        \vdots \\
        \bar \gamma^i_{\numlevels^i} \notag
    \end{bmatrix} \\
     &= \Lagrangian_{k-1}^i(\var^i, \var^{\neg i}, \eta^i_{k-1:\numlevels^i}). \notag
\end{align}
% With this construction, every term appearing in the reduced Lagrangian \(\Lagrangian^i_{k-1}\) agrees with the corresponding term in the complete Lagrangian \(\bar\Lagrangian^i_{k-1}\), after deleting those terms associated with omitted lower-level dual stationarity conditions and omitted complementarity conditions. Equivalently, the reduced stationarity system \(\mathcal F^i_{k-1}\) is obtained by restricting \(\bar{\mathcal F}^i_{k-1}\) to the subset of equations retained in the reduced formulation. Since the complete-system solution satisfies all equations of \(\bar{\mathcal F}^i_{k-1}\), it in particular satisfies all equations of \(\mathcal F^i_{k-1}\).
With this contruction, the reduced system ${\mathcal{F}}^i_{k-1}$ is satisfied by the complete-system solution $(\bar\var, \bar\eta^i_{k:\numlevels^i})$.
The reduced system \(\mathcal G^i_{k-1}\) retains only the nonnegativity conditions of $\bar\gamma^i_{k-1,1}$ and inherits the conditions
$
\bar\gamma^i_{k,1}\ge \bm{0}, \dots, \bar\gamma^i_{\numlevels^i,1}\ge \bm{0}
$
from level \(k\), which hold by the induction hypothesis. Hence \(\mathcal G^i_{k-1}\) is also satisfied.
Since the construction is valid for every player \(i\in[\numplayers]\), we conclude that every solution of the complete KKT system induces a solution of the reduced KKT system at the same primal solution. 
% Consequently, the reduced KKT system is a relaxation of the complete KKT system, and every primal solution of the complete KKT system is also a primal solution of the reduced KKT system 
% \david{don't need the last half of the sentence - it's totally redundant with the previous sentence}.
\end{proof}

\begin{proof}[Proof of \cref{thm:quadratic-growth-reduced-kkt}]
    Fix a player $i\in[\numplayers]$. The total number of variables in the reduced KKT system $(\mathcal{F}^i, \mathcal{G}^i)$ consists of (i) primal variables $\var^i$ and (ii) dual variables $\eta^i_{1:\numlevels^i}$. 
    For any $k \leq \numlevels^i$, the dimension of the dual vector $\eta^i_k$ is $(\numlevels^i-k)n^i + (\numlevels^i-k+1)\inequaldim^i+\equaldim^i$. 
    The sum of all variables' dimensions is then 
    \begin{align}
        N_{\var^i,\eta^i} & = n^i + \sum_{k=1}^{\numlevels^i} (\numlevels^i-k)n^i + (\numlevels^i-k+1)\inequaldim^i+\equaldim^i \\
        & = \Bigl(1 + \frac{\numlevels^i(\numlevels^i-1)}{2}\Bigr)n^i
        + \numlevels^i\,\equaldim^i + \frac{\numlevels^i(\numlevels^i+1)}{2}\,\inequaldim^i. \notag
    \end{align}
\end{proof}

\begin{proof}[Proof of \cref{thm:ng-og-primal-equivalence}]
The proof has three parts: 
(a) we establish the column-space inclusion relation between the reduced and complete linear systems,
(b) we propagate this relation across levels using the recursive block structure of the two KKT matrices, and 
(c) we use the resulting column-space inclusions to show that the primal solution sets of the two systems coincide. 
% \david{give them explicit abc to match below}

(a) We first prove the equivalent statements: 
        $\text{Col}(R_{k}) \subseteq \text{Col}(\bar{R}_{k}) \Longleftrightarrow \mathcal N(\bar R_{k}^\top) \subseteq \mathcal N(R_{k}^\top).
            \label{eq:kp1-level-induction-statement}
        $
        Let $y = [u^\top, v^\top, w^\top, z^\top]^\top$
    %     $\begin{pmatrix}
    %     u \\ 
    %     v \\
    %     w \\ 
    %     z
    % \end{pmatrix} 
$\in \mathcal N(\bar{R}_k^\top)$ where $u \in \mathbb{R}^n$.
Then, 
\begin{align}
    \bar{R}_k^\top y = \begin{bNiceArray}[margin]{cc|cc}
\hat Q_k & \bm{0} & \Block{2-2}{\bar R_{k+1}}  \\
       \bm{0} & \bm{0} &                   \\
\hline
\Block{2-2}{\bar R_{k+1}} & & \bm{0} & \bm{0}   \\
& & \bm{0} & \bm{0}
\end{bNiceArray} \begin{bNiceArray}{c}[margin]
        u \\ 
        v \\
        \hline
        w \\ 
        x
    \end{bNiceArray} = \begin{bmatrix}
        \begin{bmatrix}
            \hat{Q}_{k} u \\ 
            \bm{0}
        \end{bmatrix} + \bar{R}_{k+1} \begin{bmatrix}
            w \\
            x
        \end{bmatrix} \\
        \bar{R}_{k+1} \begin{bmatrix}
            u \\ v
        \end{bmatrix}
    \end{bmatrix} = \begin{bmatrix}
        \bm{0} \\
        \bm{0} \\
        \bm{0} \\
        \bm{0}
    \end{bmatrix}.
    \label{eq:preceding-line-1}
\end{align}
Premultiplying the first block row in \cref{eq:preceding-line-1} by $[u^\top, v^\top]$ and using the symmetry of $\bar R_{k+1}$: 
\begin{align}
    u^\top \hat{Q}_{k} u + \begin{bmatrix}
        u \\ v
    \end{bmatrix}^\top \bar{R}_{k+1} \begin{bmatrix}
        w \\ x
    \end{bmatrix} = u^\top \hat{Q}_{k} u + \cancel{\left(\bar{R}_{k+1} \begin{bmatrix}
        u \\ v
    \end{bmatrix}\right)}^\top \begin{bmatrix}
        w \\ x
    \end{bmatrix} = u^\top \hat{Q}_{k} u = \bm{0}.
\end{align}
Since $\hat{Q}_{k}$ is positive semi-definite, it follows that $\hat{Q}_{k} u = \bm{0}$.
Now, consider 
    \begin{align}
    R_{k}^\top y
    = \begin{bNiceArray}{c|ccc}[margin]
           \hat{Q}_{k}  & \Block{2-1}{\bm{0}} & \Block{2-1}{\bm{0}} & \Block{2-1}{\bm{0}} \\
           \bar R_{k+1,[:,1:n]}^\top & &
           \CodeAfter
            \UnderBrace[yshift=3pt]{2-1}{2-1}{\bar R_{k,[:,1:n]}^\top}
       \end{bNiceArray}
       \begin{bNiceArray}{c}[margin]
        u \\
        \hline
        v \\
        w \\ 
        z
    \end{bNiceArray} 
    = \begin{bmatrix} \hat{Q}_{k} u \\ \bar R_{k+1,[:,1:n]}^\top u \end{bmatrix} 
     = \begin{bmatrix} \bm{0} \\ \bm{0} \end{bmatrix}, \notag
     \end{align}
    where the second row follows from $\bar R_{k+1} \begin{bmatrix}
        u \\ v
    \end{bmatrix} = \bm{0}$ in \cref{eq:preceding-line-1} and the symmetry of $\bar R_{k+1}$.
    Hence, $y\in\mathcal N(R_{k}^\top)$, which implies $\mathcal N(\bar R_{k}^\top)\subseteq \mathcal N(R_{k}^\top)$. 
\smallskip

(b) We prove this result by induction on $k$.
\paragraph{\textbf{Base case} $k=\numlevels-1$} We partition the matrix $M_{\numlevels-1}$ as follows.
\begin{align}
M_{\numlevels-1} = \begin{bNiceArray}{c|ccc}[margin]
    Q_{\numlevels-1} & \bm{0} & \hat{Q}_{\numlevels} & \hat H^\top \\
    \bm{0}       & \bm{0} & \bm{0}   & \bm{0}      \\
    Q_{\numlevels}     & \hat H^\top & \bm{0} & \bm{0}   \\
    H       & \bm{0} & \bm{0}  & \bm{0} 
\CodeAfter
    \UnderBrace[yshift=3pt]{4-1}{4-1}{C_{\numlevels-1}}
    \UnderBrace[yshift=3pt]{4-2}{4-4}{M_{\numlevels-1}^c}
\end{bNiceArray} \Rightarrow M_{\numlevels-1}^c =  \begin{bNiceArray}{c|cc}[margin]
\bm{0} & \Block{2-2}{R_\numlevels} & \\
\bm{0} & & \\
\hline\\[-2.0ex]
\hat H^\top & \bm{0} & \bm{0}\\
\bm{0} &  \bm{0} & \bm{0}
\end{bNiceArray}.
\end{align}
\bigskip

Similarly, we partition the matrix $\bar M_{\numlevels-1} $ as
\begin{align}
\bar M_{\numlevels-1} = \begin{bNiceArray}{c|ccc}[margin]
    Q_{\numlevels-1} & \bm{0} & \hat Q_{\numlevels} & \hat H^\top  \\
    \bm{0}       & \bm{0} & \hat H   & \bm{0}      \\
    Q_{\numlevels}     & \hat H^\top & \bm{0} & \bm{0}   \\
    H       & \bm{0} & \bm{0}  & \bm{0} 
\CodeAfter
    \UnderBrace[yshift=3pt]{4-1}{4-1}{C_{\numlevels-1}}
    \UnderBrace[yshift=3pt]{4-2}{4-4}{\bar M_{\numlevels-1}^c}
\end{bNiceArray} \Rightarrow \bar M_{\numlevels-1}^c =  \begin{bNiceArray}{c|cc}[margin]
\bm{0} & \Block{2-2}{\bar R_\numlevels} & \\
\bm{0} & & \\
\hline\\[-2.0ex]
\hat H^\top & \bm{0} & \bm{0}\\
\bm{0} & \bm{0} & \bm{0}
\end{bNiceArray}.
\end{align}
\bigskip

By (a), $\mathrm{Col}\big(R_{\numlevels}\big) \subseteq \mathrm{Col}\big(\bar R_{\numlevels}\big)$. This implies that $\mathrm{Col}\big(M_{\numlevels-1}^c\big) \subseteq \mathrm{Col}\big(\bar M_{\numlevels-1}^c\big)$.
\paragraph{\textbf{Induction step}} Assume $\mathrm{Col}\big(M_{k}^c\big) \subseteq \mathrm{Col}\big(\bar M_{k}^c\big)$ for any $k \in [\numlevels-1]$.
Consider the level $k-1$. 
Partitioning the matrices $M_{k-1}$ and $\bar
M_{k-1}$ are:
% \david{does this require proof? where does it come from?}
\begin{align}
        M_{k-1}^c = \begin{bmatrix}
            \bm{0} & R_{k} \\
            M_k^c & \bm{0}
        \end{bmatrix}, \quad \bar M_{k-1}^c = \begin{bmatrix}
            \bm{0} & \bar R_{k} \\
            \bar M_k^c & \bm{0}
        \end{bmatrix}. 
\end{align}
Define the subspaces:
$
    A = \operatorname{Col}\!\left(\begin{bmatrix} \bm{0} \\ M_k^c\end{bmatrix}\right) 
    ,
    \bar A = \operatorname{Col}\!\left(\begin{bmatrix} \bm{0} \\ \bar M_k^c\end{bmatrix}\right) 
    ,
    $
$
    B = \operatorname{Col}\!\left(\begin{bmatrix} R_{k} \\ \bm{0}\end{bmatrix}\right),
    $ and 
$
    \bar B = \operatorname{Col}\!\left(\begin{bmatrix} \bar R_{k} \\ \bm{0}\end{bmatrix}\right) 
    . 
$
By the induction hypothesis, $\mathrm{Col}\big(M_{k}^c\big) \subseteq \mathrm{Col}\big(\bar M_{k}^c\big) \Rightarrow A \subseteq \bar A$.
By (a), $B \subseteq \bar B$. 
% \jingqi{Notation and no footnote}
% Since $\operatorname{Col}\left( M_{k-1}^c\right) = A \oplus^\perp B$ \footnote{$A$ and $B$ are orthogonal subspaces, i.e., $A \perp B$ and they intersect only at zero, i.e., $A \cap B = \{\bm{0}\}$.} and $\operatorname{Col}\left(\bar M_{k-1}^c\right) = \bar A \oplus^\perp \bar B$, we conclude that $\mathrm{Col}\big(M_{k-1}^c\big) \subseteq \mathrm{Col}\big(\bar M_{k-1}^c\big)$.
Since $A$ and $B$ are orthogonal subspaces that span $\operatorname{Col}\left( M_{k-1}^c\right)$ and, similarly, $\bar A$ and $\bar B$ are orthogonal subspaces that span $\operatorname{Col}\left( \bar M_{k-1}^c\right)$, we conclude $\mathrm{Col}\big(M_{k-1}^c\big) \subseteq \mathrm{Col}\big(\bar M_{k-1}^c\big)$.

(c) We now use the results from parts (a) and (b) to show primal solution equivalence.  
Write $\bar v_k = \begin{bmatrix}\bar \var\\ \bar\eta_{k:\numlevels}\end{bmatrix}$ and $v_k = \begin{bmatrix}\var\\ \eta_{k:\numlevels}\end{bmatrix}$ 
and define the \emph{primal} solution sets
\begin{align}
\bar{\mathcal{X}}{(\bar p_k)}
=
\Bigl\{\,\bar \var \;\Big|\; \exists\,\bar\eta_k:\;
\bar M_k \bar v_k = \bar p_k \Bigr\}, \;\;
\mathcal{X}{(\bar p_k)}
=
\Bigl\{\,\var \;\Big|\; \exists\,\eta_k:\;
M_k v_k = \bar p_k \Bigr\}.
\end{align}

(i) ${\mathcal{X}}{(\bar p_k)} \subseteq \bar{\mathcal{X}}{(\bar p_k)}.$ 
% \david{did you mean p not q?}
Partition the matrices $\bar M_k, M_k$
and let $\var \in {\mathcal{X}}{(\bar p_k)}$ so that 
\begin{align}
    M_k \begin{bmatrix}
    \var \\ \eta_k
\end{bmatrix} = \big[ c_k ~|~ M_k^c\big] \begin{bNiceArray}{c}[margin]
    \var \\ \hline \eta_k
\end{bNiceArray} = c_k \var + M_k^c \eta_k.
\end{align}
Set $\bar \var = \var$. 
Choose a vector $\bar \eta_k$ such that $M_k^c \eta_k = \bar M_k^c \bar \eta_k$, which exists by (b).
Then,
\begin{align}
    \bar M_k \begin{bmatrix}
    \bar \var \\ \bar \eta_k
\end{bmatrix} = \big[ c_k ~|~ \bar M_k^c\big] \begin{bNiceArray}{c}[margin]
    \bar \var \\ \hline \bar \eta_k
\end{bNiceArray} = \underbrace{c_k\bar \var}_{c_k\var} + \underbrace{\bar M_k^c \bar \eta_k}_{M_k^c \eta_k}.
\end{align}
This implies that $\var \in \bar {\mathcal{X}}{(\bar p_k)}$ and thus ${\mathcal{X}}{(\bar p_k)}\ \subseteq \bar{\mathcal{X}}{(\bar p_k)}$.

(ii) $\bar {\mathcal{X}}{(\bar p_k)} \subseteq {\mathcal{X}}{(\bar p_k)}.$
We proceed via induction on $k$. 

\paragraph{\textbf{Base case} $k=\numlevels-1$} Let $\bar \var \in \bar {\mathcal{X}}{(\bar p_{\numlevels-1})}$. Recalling the right-hand-side vector $\bar p_{K-1}$ in \cref{eq:complete-kkt-quadratic-goop}, we write
\begin{align}
    \bar M_{\numlevels-1} \bar v_{\numlevels-1} = \begin{bmatrix}
    Q_{\numlevels-1} & \bm{0} & \hat Q_{\numlevels} & \hat H^\top \\
    \bm{0}       & \bm{0} & \hat H   & \bm{0}      \\
    Q_{\numlevels}     & \hat H^\top & \bm{0} & \bm{0}   \\
    H       & \bm{0} & \bm{0}  & \bm{0} 
\end{bmatrix} \begin{bmatrix}
        \bar \var \\
        \bar \lambda_K \\
        \bar \psi_{K-1} \\
        \bar \lambda_{K-1}
    \end{bmatrix} = \begin{bmatrix}
        -q_{K-1} \\ 
        \bm{0} \\
        -q_K \\
        b
    \end{bmatrix} = \bar p_{\numlevels-1}.
\end{align}
% \david{not sure i would be using $:=$ since this is not a definition}
Set $v_{\numlevels-1} = \bar v_{\numlevels-1}$ (so $\var = \bar \var$). 
Then, 
\begin{align}
    M_{\numlevels-1} \bar v_{\numlevels-1} = \begin{bmatrix}
    Q_{\numlevels-1} & \bm{0} & \hat Q_{\numlevels} & \hat H^\top \\
    \bm{0}       & \bm{0} & \bm{0}   & \bm{0}      \\
    Q_{\numlevels}     & \hat H^\top & \bm{0} & \bm{0}   \\
    H       & \bm{0} & \bm{0}  & \bm{0} 
\end{bmatrix} \begin{bmatrix}
        \bar \var \\
        \bar \lambda_K \\
        \bar \psi_{K-1} \\
        \bar \lambda_{K-1}
    \end{bmatrix} = \begin{bmatrix}
        -q_{K-1} \\ 
        \bm{0} \\
        -q_K \\
        b
    \end{bmatrix} = \bar p_{\numlevels-1}
\end{align}
so $\bar \var \in \mathcal{X}(\bar p_{\numlevels-1})$.

\paragraph{\textbf{Induction step}} Assume the following holds for $k \in [\numlevels-1]$
\begin{align}
    \bar M_k \bar v_k = \bar p_k \Rightarrow M_k \bar v_k = \bar p_k.
\end{align}
From \cref{eq:complete-kkt-quadratic-goop}, the first $n$ rows are
\begin{align}
    \label{eq:first-n-rows-theorem}
    Q_{k-1} \bar \var + \bar R_{k,[1:n,:]} \bar \eta_{k-1} = q_{k-1}. 
\end{align}
Set $v_{k-1} = \bar v_{k-1}$ (so $\var = \bar \var$) 
% \david{same comment as above}.
Now consider $M_{k-1} \bar v_{k-1}$. 
From \cref{eq:reduced-kkt-quadratic-goop} using \cref{eqn:R_{k+1},eq:first-n-rows-theorem},
\begin{equation}
    \begin{bNiceArray}[margin]{cc|c}
    Q_{k-1} & \bm{0} & \Block{2-1}{R_{k}}  \\
           \bm{0} & \bm{0} &                   \\
    \hline
    \Block{2-2}{ M_{k}} &  & \Block{2-1}{\bm{0}}   \\
    &  &    
    \CodeAfter
    \UnderBrace[ , yshift=4pt]{4-1}{4-3}{M_{k-1}}
    \end{bNiceArray} 
        \begin{bNiceArray}[margin]{c}
        \bar \var \\
        \bar \eta_{\numlevels:k} \\[1ex]
        \hline\\[-2ex]
        \bar \eta_{k-1}
    \CodeAfter
    \UnderBrace[ , yshift=10pt]{3-1}{3-1}{\bar v_{k-1}}
    \end{bNiceArray} =         
    \begin{bNiceArray}[margin]{c}
        q_{k-1} \\
        \bm{0} \\
        \hline \\[-2ex]
         q_{k}  
    \CodeAfter
    \UnderBrace[ , yshift=4pt]{4-1}{4-1}{\bar p_{k-1}}
    \end{bNiceArray},
\bigskip
\medskip
\end{equation}
where the last row follows from the induction hypothesis and $v_k = [\var^\top, (\eta_{K:k})^\top]^\top$. 
Thus, $\bar \var \in {\mathcal{X}}{(\bar p_k)}$.
Combining (i) and (ii) gives $ \bar{\mathcal{X}}{(\bar p_k)} = {\mathcal{X}}{(\bar p_k)}$ for all $ k\in[\numlevels-1]$.
\end{proof}

\begin{proof}[Proof of \cref{thm:ng-og-primal-equivalence-inequality}]

The proof proceeds in three steps. %(\cref{fig:equivalence_chain}).
We first fix the active inequality constraints at the candidate primal solution \(z^*\) and show that, at \(z^*\), the corresponding inequality-constrained problem can be reduced to an equality-constrained problem on the active set.
We then apply the equality-constrained primal equivalence result from \cref{thm:ng-og-primal-equivalence}.
Finally, we reconstruct complete-system inequality multipliers from the equality-constrained multipliers using strict complementarity. 

\paragraph{\textbf{Step 1}}
Local active-set reduction. Fix a candidate local solution \(z^*\). 
% At \(z^*\), for each player $i$, we compare the original inequality-constrained problem with the equality-constrained problem obtained by replacing the active inequalities with equalities $G^i_{\mathcal{A}^i}\var = g_{\mathcal{A}^i}^i$.
% The feasible set \cref{eqn:Z^i-feasible-set} of the corresponding equality-constrained problem is then defined by $H^i\var = h^i$ and $G^i_{\mathcal{A}^i}\var = g^i_{\mathcal{A}^i}$. 
Because the active set is fixed at \(z^*\), the solution set of the inequality-constrained problem \cref{eqn:goop-K-level} for player $i$ is identical to the solution set of the equality-constrained problem obtained by replacing $G^i\var \geq g^i$ with $G^i_{\mathcal{A}^i}\var = g_{\mathcal{A}^i}^i$ and discarding the constraints that are inactive (slack) at the optimum. 
Let $\mathcal{P}_{\mathcal{I}}$ denote inequality-constrained optimization problem \cref{eqn:goop-K-level} and $\mathcal{P}_{\mathcal{E}}$ denote the corresponding equality-constrained problem.

\paragraph{\textbf{Step 2}} The equality-constrained problem $\mathcal{P}_{\mathcal{E}}$ shares the same primal solution set as its reduced KKT system.
This follows directly from \cref{thm:ng-og-primal-equivalence} applied to $\mathcal{P}_{\mathcal{E}}$, given that $\tilde H$ satisfies the required full row rank assumption.

\paragraph{\textbf{Step 3}} The reduced KKT system for the equality-constrained problem ($\mathcal{E}$) and the reduced KKT system for the inequality-constrained problem ($\mathcal{I}$) share the same set of primal solutions. Consider the stationarity conditions for level $k$ for both systems.    

\textbf{Reduced KKT System for $\mathcal{E}$:} 
Let $\mu^{\mathcal{E}} \coloneqq (\var^{\mathcal{E}},  \lambda^{\mathcal{E}}, \psi^{\mathcal{E}})$. 
% Variables are $\mu^{i,\mathcal{E}} \coloneqq (\var^{i,\mathcal{E}},  \lambda^{i,\mathcal{E}}, \psi^{i,\mathcal{E}})$. 
For $i \in [\numplayers]$ and $k \in[\numlevels^i]$:
\begin{equation}
\label{eq:proof-step3-E}
\sum_{j=1}^{\numplayers}Q^{i,i,j}_{k}\var^{j,\mathcal{E}} - \sum_{\ell=1}^{\numlevels^i-k} \Big[Q^{i,i,i}_{\numlevels^i-\ell+1}\psi_{k,\ell}^{i,\mathcal{E}} \Big] - H_i^{i\top}\lambda_{k,1}^{i,\mathcal{E}} - G^{i\top}_{i,\mathcal{A}^i}\lambda_{k,2}^{i,\mathcal{E}} = -q^i_{k},
\end{equation}
subject to $H^i\var^\mathcal{E} =   h^i$ and $G^i_{\mathcal{A}^i}\var^\mathcal{E} =   g^i_{\mathcal{A}^i}$.

\textbf{Reduced KKT System for $\mathcal{I}$:} Let $\mu^{\mathcal{I}} \coloneqq (z^{\mathcal{I}}, \lambda^{\mathcal{I}},\psi^{\mathcal{I}}, \phi^{\mathcal{I}}, \gamma^{\mathcal{I}})$. For $i \in [\numplayers]$ and $k \in[\numlevels^i]$:
\begin{align}
\label{eq:proof-step3-I}
\sum_{j=1}^{\numplayers}Q^{i,i,j}_{k}\var^{j,\mathcal{I}} - \sum_{\ell=1}^{\numlevels^i-k} \Big[Q^{i,i,i}_{\numlevels^i-\ell+1}\psi_{k,\ell}^{i,\mathcal{I}} &+ G_{i,\mathcal{A}^i}^{i\top} \left( \phi_{k,\ell}^{i,\mathcal{I}} \odot \gamma_{\numlevels^i-\ell+1}^{i,\mathcal{I}}\right) \Big] \\ &- H_i^{i\top}\lambda_{k}^{i,\mathcal{I}} - G^{i\top}_{i,\mathcal{A}^i}\gamma_{k}^{i,\mathcal{I}} = -q^i_{k}, \notag
\end{align}
subject to $H^i\var^\mathcal{I} =   h^i$, $G^i_{\mathcal{A}^i} \var^\mathcal{I} \geq  g^i_{\mathcal{A}^i}$, $\gamma^{i,\mathcal{I}}_{1:\numlevels^i} \ge \bm{0}$, and complementarity conditions.

\textbf{3.1} If $\mu^{\mathcal{I}}$ satisfies ($\mathcal{I}$), then there exists $\mu^{\mathcal{E}}$ satisfying ($\mathcal{E}$) such that $\var^{\mathcal{E}} = \var^{\mathcal{I}}$.
Assume $\mu^{\mathcal{I}} \coloneqq (z^{\mathcal{I}},\lambda^{\mathcal{I}},\psi^{\mathcal{I}}, \phi^{\mathcal{I}}, \gamma^{\mathcal{I}})$ satisfies ($\mathcal{I}$).
Define $(\mathcal{E})$ variables for $i \in[\numplayers]$:
$
    \var^{i,\mathcal{E}} \coloneqq \var^{i,\mathcal{I}}, \;
    \lambda_{k,1}^{i,\mathcal{E}} \coloneqq \lambda_k^{i,\mathcal{I}},
    \;
    \psi^{i,\mathcal{E}} \coloneqq \psi^{i,\mathcal{I}}, \;
    \lambda_{k,2}^{i,\mathcal{E}} \coloneqq \gamma_k^{i,\mathcal{I}} + \sum_{\ell=1}^{\numlevels^i-k} \phi_{k,\ell}^{i,\mathcal{I}} \odot \gamma_{\numlevels^i-\ell+1}^{i,\mathcal{I}}.
$
Substituting these definitions into \cref{eq:proof-step3-I} yields stationarity in \cref{eq:proof-step3-E}.
Next, verify feasibility. We have $H^i\var^\mathcal{I} =   h^i.$
By the construction of the active set \(\mathcal A^i\), any solution to ($\mathcal{I}$) must satisfy $G^i_{\mathcal{A}^i}\var^\mathcal{I} =   g^i_{\mathcal{A}^i}$. Thus $\var^\mathcal{E}$ satisfies the equality constraints of ($\mathcal{E}$).

\textbf{3.2} If $\mu^\mathcal{E}$ satisfies ($\mathcal{E}$), then there exists $\mu^\mathcal{I}$ satisfying ($\mathcal{I}$) such that $\var^\mathcal{I} = \var^\mathcal{E}$.
Assume $\mu^\mathcal{E} \coloneqq (\var^\mathcal{E},\psi^\mathcal{E},\lambda^\mathcal{E})$ satisfies ($\mathcal{E}$). We construct a solution for ($\mathcal{I}$) by induction on the level index $k$, proceeding backwards from $\numlevels^i$ to $1$.
First, set shared variables for $i\in[\numplayers]$:
    $
        \var^{i,\mathcal{I}} \coloneqq \var^{i,\mathcal{E}}, 
        \;
        \lambda_k^{i,\mathcal{I}} \coloneqq \lambda_{k,1}^{i,\mathcal{E}},
        \;
        \psi^{i,\mathcal{I}} \coloneqq \psi^{i,\mathcal{E}}.
    $
    Since $H^i\var^\mathcal{E}=  h^i$ and $G^i_{\mathcal{A}^i}\var^\mathcal{E}=  g^i_{\mathcal{A}^i}$, $\var^\mathcal{I}$ satisfies primal feasibility with respect to the innermost constraints.

\textbf{3.2.1 \emph{Base Case}} ($k=\numlevels^i$).
Fix a player $i$. Define $\gamma_{\numlevels^i}^{i,\mathcal{I}} = \lambda_{\numlevels^i}^{i,\mathcal{E}}$, where \(\lambda_{\numlevels^i}^{i,\mathcal E}\) denotes the multipliers associated with the equality constraints obtained from the active inequalities at level \(\numlevels^i\). By strict complementarity at the candidate solution, these multipliers are strictly positive. Hence $\gamma_{\numlevels^i}^{i,\mathcal I} > \bm{0}$.
% so the nonnegativity and complementarity conditions at level \(K\) are satisfied.

\textbf{3.2.2 \emph{Induction Step}}
Assume that, for all levels ${m^i} > k$, we have determined $\gamma_{m^i}^{i,\mathcal{I}} \ge \bm{0}$ and $\phi_{m^i,\cdot}^{i,\mathcal{I}}$. We now determine $\gamma_k^{i,\mathcal{I}} \ge \bm{0}$ and $\phi_{k,\cdot}^{i,\mathcal{I}}$ for \cref{eq:proof-step3-I}.
From \cref{eq:proof-step3-I} and \cref{eq:proof-step3-E}, matching the terms requires
$
    \lambda_{k,2}^{i,\mathcal{E}} = \gamma_k^{i,\mathcal{I}} + \sum_{\ell=1}^{\numlevels^i-k} \phi_{k,\ell}^{i,\mathcal{I}} \odot \gamma_{\numlevels^i-\ell+1}^{i,\mathcal{I}}.
$
To satisfy this, it suffices to choose one lower-level term and set the remaining \(\phi_{k,\ell}^{i,\mathcal I}\) equal to zero. We use the term corresponding to \(\ell=1\), which involves \(\gamma_{\numlevels^i}^{i,\mathcal I}\), already known to be strictly positive from the base case.
Set $\phi_{k,\ell}^{i,\mathcal{I}} = \bm{0}$ for all $\ell > 1$.
Then, we have
$
    \lambda_{k,2}^{i,\mathcal{E}} = \gamma_k^{i,\mathcal{I}} + \phi_{k,1}^{i,\mathcal{I}} \odot \gamma_{\numlevels^i}^{i,\mathcal{I}}.
$
For each component $r \in \{1, \dots, \inequaldim^i\}$:
\begin{itemize}
    \item If $(\lambda_{k,2}^{i,\mathcal{E}})_r \geq 0$: 
        Set $(\gamma_k^{i,\mathcal{I}})_r \coloneqq (\lambda_{k,2}^{i,\mathcal{E}})_r$ and $(\phi_{k,1}^{i,\mathcal{I}})_r \coloneqq 0$.
    \item If $(\lambda_{k,2}^{i,\mathcal{E}})_r < 0$: 
        Set $(\gamma_k^{i,\mathcal{I}})_r \coloneqq 0$. We must solve $(\lambda_{k,2}^{i,\mathcal{E}})_r = (\phi_{k,1}^{i,\mathcal{I}})_r \cdot (\gamma_{\numlevels^i}^{i,\mathcal{I}})_r$. 
        Since $(\gamma_{\numlevels^i}^{i,\mathcal{I}})_r > 0$ (from Base Case 3.2.1), we can uniquely define:
        $
        (\phi_{k,1}^{i,\mathcal{I}})_r \coloneqq \frac{(\lambda_{k,2}^{i,\mathcal{E}})_r}{(\gamma_{\numlevels^i}^{i,\mathcal{I}})_r}.
        $
\end{itemize} This construction ensures $\gamma_k^{i,\mathcal{I}} \geq \bm{0}$, satisfying nonnegativity.
The complementarity condition $(G^i_{\mathcal{A}^i}\var^{i,\mathcal{E}} - g_{\mathcal{A}^i}) \odot \gamma_k^{i,\mathcal{I}} = \bm{0}$ holds naturally, because primal feasibility gives $G^i_{\mathcal{A}^i}\var^{i,\mathcal{E}} = g_{\mathcal{A}^i}$. 
Repeat steps 3.2.1 and 3.2.2 for all players $i \in [\numplayers]$.
Thus, we obtain $\mu^\mathcal{I}$ satisfying ($\mathcal{I}$) such that $\var^\mathcal{I} = \var^\mathcal{E}$.
Combining Steps 1, 2, and 3 proves the theorem.
\end{proof}

\begin{proof}[Proof of Theorem~\ref{thm:quadratic-convergence}]
    Let $y_\rho\in\mathcal{S}_{y_\rho}$. We begin the proof by characterizing how the pseudoinverse $(\nabla \mathcal{K}_\rho(y_\rho))^+$ affects the Newton update. 
    Since $\mathcal{K}_\rho(y_\rho)$ is in the column space of $\nabla \mathcal{K}_\rho(y_\rho)$ for all $y_\rho\in\mathcal{S}_{y_\rho}$, 
    % we have for all vectors $u$ in the left null space of $\nabla \mathcal{K}_\rho(y_\rho)$, and by the fundamental theorem of calculus, 
    % % \david{should mention what $y^*$ is}
    % \begin{equation}
    % \begin{aligned}
    %     u^\top \mathcal{K}_\rho(y_\rho) & = u^\top \underbrace{\mathcal{K}_\rho(y^*_\rho)}_{=0} +  u^\top \int_{0}^1 \nabla \mathcal{K}_\rho(y^*_\rho + t(y_\rho - y^*_\rho))^\top (y_\rho - y^*_\rho) dt = 0.
    %     % & = 0 + \int_0^1 u^\top \nabla \mathcal{K}_\rho(y^*_\rho + t(y_\rho - y^*_\rho)) dt \cdot (y_\rho - y^*_\rho) = 0
    % \end{aligned}
    % \end{equation}
    % Thus, the vector $\mathcal{K}_\rho(y_\rho)$ is orthogonal to $\textrm{ker}(\nabla \mathcal{K}_\rho (y_\rho)^\top)$ and, thus, lies in the column space of $\nabla \mathcal{K}_\rho(y_\rho)$, and 
    the variable $\Delta y_\rho= (\nabla \mathcal{K}_\rho(y_\rho))^+ ($ $-\mathcal{K}_\rho(y_\rho))$ satisfies the equation $        \nabla \mathcal{K}_\rho (y_\rho) \Delta y_\rho = -\mathcal{K}_\rho(y_\rho)$, 
    % \begin{equation}\label{eq: newton update under pseudo inverse}
    %     \nabla \mathcal{K}_\rho (y) \Delta y = -\mathcal{K}_\rho(y). 
    % \end{equation}
    % Since $\nabla \mathcal{K}_\rho(y)$ has constant row rank for all $ y\in \mathcal{S}_y$, the projector onto the column space of $\nabla \mathcal{K}_\rho(y)$, 
    % \begin{equation}
    %     P(y):= \nabla \mathcal{K}_\rho(y) (\nabla \mathcal{K}_\rho(y))^+
    % \end{equation}
    % varies continuously and has constant rank throughout $\mathcal{S}_y$. Since a solution $y^* \in \mathcal{S}_y$ with $\mathcal{K}_\rho(y^*)=0$ exists, $P(y^*)\mathcal{K}_\rho(\mathbf{y^*})=\mathcal{K}_\rho(y^*)$ holds trivially. The constant rank condition ensures $P(y)$ is a smooth projector over $\mathcal{S}_y$. By continuity, for all $y\in\mathcal{S}_y$, we have 
    % \begin{equation}
    %     P(y) \mathcal{K}_\rho(y) = \mathcal{K}_\rho(y)
    % \end{equation}
    % which implies that $\mathcal{K}_\rho(y)$ lies in the column space of $\nabla \mathcal{K}_\rho(y)$ throughout $\mathcal{S}_y$. This also implies that 
    % \begin{equation}\label{eq: newton update under pseudo inverse}
    %     \nabla \mathcal{K}_\rho(y) \Delta y = -P(y) \mathcal{K}_\rho(y) = -\mathcal{K}_\rho(y)
    % \end{equation}
    where the complementarity rows of this equation give for each $i^\mathrm{th}$ element of the vectors of $\{s_\rho,\gamma_\rho, \Delta s_\rho, \Delta \gamma_\rho\}$:
    \begin{equation}\label{eq:convergence proof complementarity conditions}
        s_{\rho,i} \Delta \gamma_{\rho,i} + \gamma_{\rho,i} \Delta s_{\rho,i} = \rho - \gamma_{\rho,i} s_{\rho,i}.
    \end{equation}
    
    We then proceed to prove the linear convergence property in \eqref{eq: linear convergence}. It suffices to show the existence of a stepsize $\alpha\in(0,1]$ satisfying \eqref{eq: linear convergence} by showing that $\hat{\alpha}$ is a stepsize rendering $ \|\mathcal{K}_\rho(y_\rho+\hat{\alpha}\Delta y_
\rho)\|_2\le (1-\frac{\hat{\alpha}}{2}) \|\mathcal{K}_\rho(y_
\rho)\|_2 $. Following the fundamental theorem of calculus again, we have, for all stepsize $\alpha\in(0,1]$,
    \begin{equation}
    \begin{aligned}
        &\|\mathcal{K}_\rho(y_
\rho+ \alpha \Delta y_\rho)\|_2=\left\| \mathcal{K}_\rho(y_\rho) + \int_0^1 \nabla \mathcal{K}_\rho(y_\rho + \tau \alpha \Delta y_\rho )\alpha \Delta y_\rho d\tau  \right\|_2\\
        &\le \| \mathcal{K}_\rho(y_\rho) + \alpha \nabla K_\rho( y_\rho) \Delta y_\rho\|_2 + \left\| \int_0^1 (\nabla \mathcal{K}_\rho(y_\rho + \tau \alpha \Delta y_\rho) - \nabla \mathcal{K}_\rho(y_\rho))\alpha \Delta y_\rho d\tau \right\|_2\\
        & \le \| \mathcal{K}_\rho(y_\rho) - \alpha \mathcal{K}_\rho(y_\rho) \|_2 + \|\alpha \Delta y_\rho\|_2  \cdot \int_0^1 C \|\alpha \tau \Delta y_\rho \|_2 d\tau \\
        % & \le (1-\alpha) \|\mathcal{K}_\rho(y_\rho)\|_2  + \|\alpha \Delta y_\rho\|_2  \cdot \int_0^1 C \|\alpha \tau \Delta y_\rho \|_2 d\tau\\
        & \le (1-\alpha) \|\mathcal{K}_\rho(y_\rho)\|_2 + \frac{1}{2} \alpha^2 CD^2 \|\mathcal{K}_\rho(y_\rho)\|_2^2.
    \end{aligned}\label{eq:basic inequality}
    \end{equation}
    Let $y_\rho\in\mathcal{S}_{y_\rho}$. Since $\hat{\alpha}\le \frac{1}{CD^2L}$ and $\|\mathcal{K}_\rho(y_\rho)\|_2\le L$, we have  $\frac{C\hat{\alpha}D^2}{2}\|\mathcal{K}_\rho(y_\rho)\|_2\le \frac{1}{2}$, and 
    \begin{equation}
    \begin{aligned}
        \|\mathcal{K}_\rho(y_\rho+\hat{\alpha} \Delta y_\rho)\|_2  \le (1-\hat{\alpha})\| \mathcal{K}_\rho(y_\rho) \|_2 + \frac{1}{2} \hat{\alpha} \|\mathcal{K}_\rho(y_\rho)\|_2\le \left(1-\frac{1}{2}\hat{\alpha}\right) \|\mathcal{K}_\rho(y_\rho)\|_2 \le L.
    \end{aligned}
    \end{equation}
    which establishes \eqref{eq: linear convergence}. To complete the proof of $y_\rho+\hat{\alpha}\Delta y_\rho\in\mathcal{S}_{y_\rho}$, we need to show $(s_\rho+\hat{\alpha}\Delta s_\rho)\odot(\gamma_\rho +\hat{\alpha}\Delta \gamma_\rho)\ge \epsilon \mathbf{1}$. For every $i^\mathrm{th}$ element of the vectors $\{s_\rho, \gamma_\rho, \Delta s_\rho, \Delta \gamma_\rho\}$, using the complementarity rows in \eqref{eq:convergence proof complementarity conditions}
    and substituting the Newton's update,
    \begin{equation}
        (\gamma_{\rho,i} + \hat{\alpha}\Delta \gamma_{\rho,i}) (s_{\rho,i} + \hat{\alpha}\Delta s_{\rho,i}) = (1-\hat{\alpha})\gamma_{\rho,i} s_{\rho,i} + \hat{\alpha} \rho + \hat{\alpha}^2\Delta \gamma_{\rho,i} \Delta s_{\rho,i}.
    \end{equation}
    From $\Delta y_\rho = -(\nabla \mathcal{K}_\rho(y_\rho))^+\mathcal{K}_\rho(y_\rho)$, we have $\|\Delta y_\rho\|_2 \le D \|\mathcal{K}_\rho(y_\rho)\|_2\le D L$, and $\|\Delta \gamma_{\rho,i} \Delta s_{\rho,i}\|_2$ $ \le \|\Delta y_\rho\|_2^2 \le D^2 L^2$. Since $\gamma_{\rho,i} s_{\rho,i} \ge \epsilon$ and $\hat{\alpha}\le \frac{\rho-\epsilon}{D^2L^2}$, we have
    \begin{equation}
        (\gamma_{\rho,i} + \hat{\alpha} \Delta \gamma_{\rho,i})(s_{\rho,i} + \hat{\alpha}\Delta s_{\rho,i}) \ge \epsilon + \hat{\alpha}\left( (\rho-\epsilon) - \hat{\alpha} D^2 L^2 \right) \ge \epsilon
    \end{equation}
    which completes the proof of $y_\rho+\hat{\alpha}\Delta y_\rho \in \mathcal{S}_{y_\rho}$. 
    
    In what follows, we show the quadratic convergence of our PDIP method when $\|\mathcal{K}_\rho(y_\rho)\|_2\le \min\{\frac{2}{CD^2}, \frac{\sqrt{\rho-\epsilon}}{D}\}$. From \eqref{eq:basic inequality}, we have
    \begin{equation}
        \mathcal{K}_\rho(y_\rho+\Delta y_\rho) \le \frac{1}{2}CD^2 \|\mathcal{K}_\rho(y_\rho)\|_2^2 \le \frac{1}{2}CD^2 \frac{2}{CD^2} \|\mathcal{K}_\rho(y_\rho)\|_2 = \|\mathcal{K}_\rho(y_\rho)\|_2 \le L.
    \end{equation}
    Moreover, since $\|\Delta s_{\rho,i} \Delta \gamma_{\rho,i}\|_2 \le D^2 \|\mathcal{K}_\rho(y_\rho)\|_2^2$, we have
    \begin{equation}
        (\gamma_{\rho,i} + \Delta \gamma_{\rho,i})(s_{\rho,i} + \Delta s_{\rho,i}) \ge \epsilon + (\rho - \epsilon) - D^2 \|\mathcal{K}_\rho(y_\rho)\|_2^2 \ge \epsilon
    \end{equation}
    and therefore $y_\rho+\Delta y_\rho\in\mathcal{S}_{y_\rho}$. 
    
    Observing $ \frac{CD^2}{2} \|\mathcal{K}_\rho(y_\rho+\Delta y_\rho)\|_2 \le \left( \frac{CD^2}{2} \|\mathcal{K}_\rho(y_\rho)\|_2 \right)^2 \le 1 $, we have $\|\mathcal{K}_\rho(y_\rho+\Delta y_\rho)\|_2\le \min\{\frac{2}{CD^2}, \frac{\sqrt{\rho-\epsilon}}{D}\}$. By induction, this condition holds at all iterates. 
    % This completes the proof. 
\end{proof}

\begin{proof}[Proof of Theorem~\ref{thm:central path}]
    Let $\mathbf{1}_\rho$ be defined as a binary vector, where we have ones for the entries corresponding to the complementarity slackness condition. 
    
    Since $\mathcal{K}_\rho(y_\rho) = \mathcal{K}_0(y_\rho) - \rho \cdot \mathbf{1}_\rho$, evaluating at $y_\rho^*$, we have 
    $\mathcal{K}_0(y_\rho^*) = \rho \mathbf{1}_\rho$. 
    By the fundamental theorem of calculus, we have
    \begin{equation}\label{eq:fundamental theorem of calculas of central path}
        \mathcal{K}_0(y_\rho^*) - \mathcal{K}_0(y_0^*) = \underbrace{\int_0^1 \nabla_y \mathcal{K}_0(y_0^* + t(y_\rho^* - y_0^*))dt}_{\bar{J}} \cdot (y_\rho^* - y_0^*).
    \end{equation}
    Since $\nabla_y \mathcal{K}_0$ has constant row rank in a neighborhood of $y_0^*$, for $\rho$ sufficiently small $y_\rho^* $ lies in this neighborhood and the right hand side of \eqref{eq:fundamental theorem of calculas of central path} has the same constant row rank, with $\| \bar{J}^+ \|_2 \le D_0$ by continuity of the pseudoinverse under constant rank. 
    %Combining \eqref{eq: residual central path} and \eqref{eq:fundamental theorem of calculas of central path}, 
    Multiplying both sides of \eqref{eq:fundamental theorem of calculas of central path} by $\bar{J}^+$, we have
    \begin{equation}
        \| y_\rho^* - y_0^* \|_2 \le \|\bar{J}^+\|_2\|\rho \mathbf{1}_\rho\|_2 \le D_0 \sqrt{N_c} \rho
    \end{equation}
\end{proof}

\bibliographystyle{siamplain}
\bibliography{references.bib}

\end{document}